\documentclass[a4paper]{article}

\usepackage[T1]{fontenc}
\usepackage{pst-plot,pst-node}
\usepackage{latexsym}
\usepackage{amssymb}
\usepackage{epsfig}
\usepackage{amsfonts}
\usepackage{color}
\usepackage{amsmath}
\usepackage[all]{xy}
\usepackage{graphicx}
\usepackage{hyperref}

\newcommand{\bull}{\rule{.85ex}{1ex} \par \bigskip}
\newenvironment{proof}{\noindent {\bf Proof:\ }}{\hfill $\Box$ \par \bigskip}

\newtheorem{theorem}{Theorem}[section]
\newtheorem{definition}[theorem]{Definition}
\newtheorem{proposition}[theorem]{Proposition}
\newtheorem{lemma}[theorem]{Lemma}
\newtheorem{corollary}[theorem]{Corollary}
\newtheorem{conjecture}[theorem]{Conjecture}

\title{Hard constraint satisfaction problems\\ have hard gaps at location 1\footnotemark[1]}
\author{Peter Jonsson\footnotemark[2] \and Andrei Krokhin\footnotemark[3] \and Fredrik Kuivinen\footnotemark[4]}

\begin{document}
\maketitle

\footnotetext[1]{Preliminary versions of parts of this report appeared in {\em Proceedings of the 2nd International Computer Science Symposium in Russia (CSR-2007)}, Ekaterinburg, Russia, 2007}

\footnotetext[2]{Department of Computer and Information Science, Link\"{o}pings Universitet, SE-581 83 Link\"{o}ping, Sweden, email: {\tt petej@ida.liu.se}, phone: +46 13 282415, fax: +46 13 284499}

\footnotetext[3]{Department of Computer Science, University of Durham,
Science Laboratories, South Road, Durham DH1 3LE, UK, email: {\tt andrei.krokhin@durham.ac.uk}, phone: +44 191 334 1743}

\footnotetext[4]{Department of Computer and Information Science, Link\"{o}pings Universitet, SE-581 83 Link\"{o}ping, Sweden, email: {\tt freku@ida.liu.se}, phone: +46 13 286607, fax: +46 13 284499}

\newenvironment{example}{\noindent {\bf Example:\ } \rm }{\hfill}

\newcommand{\QCSP}[1]{\mbox{\rm QCSP$(#1)$}}
\newcommand{\MCSP}[1]{\mbox{{\sc Max CSP}$(#1)$}}
\newcommand{\wMCSP}[1]{\mbox{\rm weighted Max CSP$(#1)$}}
\newcommand{\cMCSP}[1]{\mbox{\rm cw-Max CSP$(#1)$}}
\newcommand{\tMCSP}[1]{\mbox{\rm tw-Max CSP$(#1)$}}
\renewcommand{\P}{\mbox{\bf P}}
\newcommand{\G}[1]{\mbox{\rm I$(#1)$}}
\newcommand{\NE}[1]{\mbox{$\neq_{#1}$}}

\newcommand{\NP}{\mbox{\bf NP}}
\newcommand{\NL}{\mbox{\bf NL}}
\newcommand{\PO}{\mbox{\bf PO}}
\newcommand{\NPO}{\mbox{\bf NPO}}
\newcommand{\APX}{\mbox{\bf APX}}
\newcommand{\Aut}{\mbox{\rm Aut}}
\newcommand{\bound}{\mbox{\rm -$B$}}

\newcommand{\GIF}[3]{\ensuremath{h_\{{#2},{#3}\}^{#1}}}

\newcommand{\Spmod}{\mbox{\rm Spmod}}
\newcommand{\Sbmod}{\mbox{\rm Sbmod}}

\newcommand{\Inv}[0]{\mbox{\rm Inv}}
\newcommand{\Pol}[0]{\mbox{\rm Pol}}
\newcommand{\sPol}[1]{\mbox{\rm s-Pol($#1$)}}

\newcommand{\un}{\underline}
\newcommand{\ov}{\overline}
\def\vect#1#2{#1 _1\zdots #1 _{#2}}
\def\zd{,\ldots,}
\let\sse=\subseteq
\let\la=\langle
\def\lla{\langle\langle}
\let\ra=\rangle
\def\rra{\rangle\rangle}
\let\vr=\varrho
\def\vct#1#2{#1 _1\zd #1 _{#2}}
\newcommand{\va}{{\bf a}}
\newcommand{\vb}{{\bf b}}
\newcommand{\vc}{{\bf c}}
\newcommand{\bx}{{\bf x}}
\newcommand{\by}{{\bf y}}
\def\Z{{\bur Z^+}}
\def\R{{\bur R}}
\def\D{{\cal D}}
\def\F{{\cal F}}
\def\I{{\cal I}}
\def\C{{\cal C}}
\def\U{{\cal U}}
\def\K{{\cal K}}
\def\Lat{{\cal L}}

\def\2mat#1#2#3#4#5#6#7#8{
\begin{array}{c|cc}
$~$ & #3 & #4\\
\hline
#1 & #5& #6\\
#2 & #7 & #8 \end{array}}

\font\tenbur=msbm10
\font\eightbur=msbm8
\def\bur{\fam11}
\textfont11=\tenbur \scriptfont11=\eightbur \scriptscriptfont11=\eightbur
\font\twelvebur=msbm10 scaled 1200
\textfont13=\twelvebur \scriptfont13=\tenbur \scriptscriptfont13=\eightbur




\newcounter{A}
\newcounter{prgline} 


\newcommand{\citelist}[1]{\raisebox{.2ex}{[}#1\raisebox{.2ex}{]}}
\newcommand{\scite}[1]{\citeauthor{#1}, \citeyear{#1}}
\newcommand{\shortcite}[1]{\cite{#1}}
\newcommand{\mciteii}[2]{\citeauthor{#1}, \citeyear{#1}, %
\citeyear{#2}}
\newcommand{\mciteiii}[3]{\citeauthor{#1}, \citeyear{#1}, %
\citeyear{#2}, \citeyear{#3}}

\newcommand{\multiciteii}[2]{\citelist{\scite{#1}, \citeyear{#2}}}
\newcommand{\multiciteiii}[3]%
  {\citelist{\scite{#1}, \citeyear{#2}, \citeyear{#3}}}


\renewcommand{\phi}{\varphi}
\renewcommand{\epsilon}{\varepsilon}

\newcommand{\draft}{\begin{center}\huge Draft!!! \end{center}}
\newcommand{\void}{\makebox[0mm]{}}     


\renewcommand{\text}[1]{\mbox{\rm \,#1\,}}        

\newcommand{\tand}{\text{\ and\ }}
\newcommand{\tor}{\text{\ or\ }}
\newcommand{\tif}{\text{\ if\ }}
\newcommand{\tiff}{\text{\ iff\ }}
\newcommand{\tfor}{\text{\ for\ }}
\newcommand{\tforall}{\text{\ for all\ }}
\newcommand{\totherwise}{\text{\ otherwise}}


\newcommand{\fnl}{\void \\}             

\newcommand{\pushlist}[1]{\setcounter{#1}{\value{enumi}} \end{enumerate}}
\newcommand{\poplist}[1]{\begin{enumerate} \setcounter{enumi}{\value{#1}}}


\renewcommand{\emptyset}{\varnothing}  
\newcommand{\union}{\cup}               
\newcommand{\intersect}{\cap}           
\newcommand{\setdiff}{-}                
\newcommand{\compl}[1]{\overline{#1}}   
\newcommand{\card}[1]{{|#1|}}           
\newcommand{\set}[1]{\{{#1}\}} 
\newcommand{\st}{\ |\ }                 
\newcommand{\suchthat}{\st}             
\newcommand{\cprod}{\times}             
\newcommand{\powerset}[1]{{\bf 2}^{#1}} 

\newcommand{\tuple}[1]{\langle{#1}\rangle}  
\newcommand{\seq}[1]{\langle #1 \rangle}
\newcommand{\emptyseq}{\seq{}}
\newcommand{\floor}[1]{\left\lfloor{#1}\right\rfloor}
\newcommand{\ceiling}[1]{\left\lceil{#1}\right\rceil}

\newcommand{\map}{\rightarrow}
\newcommand{\fncomp}{\!\circ\!}         

\newcommand{\transclos}[1]{#1^+}
\newcommand{\reduction}[1]{#1^-}        

\newtheorem{defnx}{Definition}
\newtheorem{axiomx}[defnx]{Axiom}
\newtheorem{theoremx}[defnx]{Theorem}
\newtheorem{propositionx}[defnx]{Proposition}
\newtheorem{lemmax}[defnx]{Lemma}
\newtheorem{corollx}[defnx]{Corollary}
\newtheorem{algx}[defnx]{Algorithm}
\newtheorem{exx}[defnx]{Example}
\newtheorem{factx}[defnx]{Fact}

\newcommand{\QED}{\nopagebreak[4]{\makebox[1mm]{}\hfill$\Box$}}


\newlength{\prgindent}
\newenvironment{program}{
  \begin{list}%
   {\arabic{prgline}}%
   {
   \usecounter{prgline}
   \setlength{\prgindent}{0em}
   \setlength{\parsep}{0em}
   \setlength{\itemsep}{0em}
   \setlength{\labelwidth}{1em}
   \setlength{\labelsep}{1em}
   \setlength{\leftmargin}{\labelwidth}
   \addtolength{\leftmargin}{\labelsep}
   \setlength{\topsep}{0em}
   \setlength{\parskip}{0em} } }%
 {\end{list}}

\newif\ifprgendtext
\prgendtexttrue

\newenvironment{prgblock}{\addtolength{\prgindent}{\labelsep}}%
{\addtolength{\prgindent}{-\labelsep}}
\newcommand{\prgbeginblock}{\addtolength{\prgindent}{\labelsep}}
\newcommand{\prgendblock}{\addtolength{\prgindent}{-\labelsep}}
\newcommand{\prgcndendblock}[1]{\addtolength{\prgindent}{-\labelsep}
 \ifprgendtext \prglin\prgres{#1}\fi}

\newcommand{\prglin}{\rm \item\hspace{\prgindent}}
\newcommand{\prgcontlin}{\\  \hspace{\prgindent}}

\newcommand{\prgres}[1]{{\bf #1}}
\newcommand{\prgassn}{\leftarrow}
\newcommand{\prgname}[1]{{\it #1}}

\newcommand{\prgbegin}{\prglin\prgres{begin}\prgbeginblock}
\newcommand{\prgend}{\prgendblock\prglin\prgres{end\ }}
\newcommand{\prgnoend}{\prgendblock}

\newcommand{\prgif}{\prglin\prgres{if\ }}
\newcommand{\prgthen}{\prgres{\ then\ }\prgbeginblock}
\newcommand{\prgelse}{\prgendblock\prglin\prgres{else}\prgbeginblock}
\newcommand{\prgelsif}{\prgendblock\prglin\prgres{elsif\ }}
\newcommand{\prgendif}{\prgcndendblock{end if}}

\newcommand{\prgwhile}{\prglin\prgres{while\ }}
\newcommand{\prgfor}{\prglin\prgres{for\ }}
\newcommand{\prgdo}{\prgres{\ do}\prgbeginblock}
\newcommand{\prgrepeat}{\prgres{\ repeat}\prgbeginblock}
\newcommand{\prgloop}{\prglin\prgres{loop}\prgbeginblock}
\newcommand{\prgendloop}{\prgendblock\prglin\prgres{end loop}}
\newcommand{\prgendwhile}{\prgcndendblock{end while}}
\newcommand{\prgendfor}{\prgcndendblock{end for}}
\newcommand{\prguntil}{\prgendblock\prglin\prgres{until\ }}

\newcommand{\prgcomment}{\prglin\prgres{comment\ }\it }
\newcommand{\prgprocedure}{\prglin\prgres{procedure\ }}
\newcommand{\prgnil}{\prgres{\ nil}}
\newcommand{\prgtrue}{\prgres{\ true}}
\newcommand{\prgfalse}{\prgres{\ false}}
\newcommand{\prgnot}{\prgres{\ not\ }}
\newcommand{\prgand}{\prgres{\ and\ }}
\newcommand{\prgor}{\prgres{\ or \ }}
\newcommand{\prgfail}{\prgres{fail}}
\newcommand{\prgreturn}{\prgres{return\ }}
\newcommand{\prgaccept}{\prgres{accept}}
\newcommand{\prgreject}{\prgres{reject}}

\prgendtextfalse


\newcommand{\ie}{{\em ie.}}                
\newcommand{\eg}{{\em eg.}}
\newcommand{\paper}{paper}                

\newcommand{\emdef}{\em}                   
\newcommand{\rinterpretation}{${\Bbb R}$-interpretation}
\newcommand{\rmodel}{${\Bbb R}$-model}
\newcommand{\transp}{^{\rm T}}

\newcommand{\unprint}[1]{}
\newcommand{\blankline}{$\:$}

\newcommand{\Solv}{{\it TSolve}}
\newcommand{\Neg}{{\it Neg}}
\newcommand{\logname}{XX}

\newcommand{\props}{{\it props}}
\newcommand{\rels}{{\it rels}}
\newcommand{\deduce}{\vdash_p}

\newcommand{\pform}{{\rm Pr}}
\newcommand{\axform}{{\rm AX}}
\newcommand{\axset}{{\bf AX}}
\newcommand{\resdeduce}{\vdash_{\rm R}}
\newcommand{\resaxdeduce}{\vdash_{\rm R,A}}

\newcommand{\cmis}{{\em \#mis}}
\newcommand{\combine}{{\em comb}}

\newcommand{\xCSP}{{\sc X-CSP}}
\newcommand{\CSP}{{\sc CSP}}

\newcommand{\Ind}[0]{\textrm{Ind}}
\newcommand{\typ}[0]{\textsf{typ}}
\newcommand{\Con}[0]{\textsf{Con}\;}

\newcommand{\lattleq}[0]{\sqsubseteq}
\newcommand{\lattgeq}[0]{\sqsupseteq}
\newcommand{\lattl}[0]{\sqsubset}
\newcommand{\lattg}[0]{\sqsupset}
\newcommand{\lub}[0]{\sqcup}
\newcommand{\glb}[0]{\sqcap}

\newcommand{\prob}[1]{{\sc #1}}
\newcommand{\opt}[0]{\textrm{{\sc opt}}}
\def\tup#1{\mathchoice{\mbox{\boldmath$\displaystyle#1$}}
{\mbox{\boldmath$\textstyle#1$}}
{\mbox{\boldmath$\scriptstyle#1$}}
{\mbox{\boldmath$\scriptscriptstyle#1$}}}

\newcommand{\strictimpl}[1]{\stackrel{s\phantom{8pt}}{\Longrightarrow_{#1}}}
\newcommand{\perfimpl}[1]{\stackrel{p\phantom{8pt}}{\Longrightarrow_{#1}}}

\newcommand{\cc}[1]{\textnormal{\textbf{#1}}} 
\newcommand{\proj}[2]{#1\big|_{#2}}
\newcommand{\ignore}[1]{}



\bibliographystyle{siam}

\begin{abstract}
  An instance of the \emph{maximum constraint satisfaction problem} (\prob{Max CSP})
  is a finite
  collection of constraints on a set of variables, and the goal is to
  assign values to the variables that maximises the number of
  satisfied constraints.  \prob{Max CSP} captures many well-known
  problems (such as \prob{Max $k$-SAT} and \prob{Max Cut}) and is
  consequently \cc{NP}-hard.  Thus, it is natural to study how
  restrictions on the allowed constraint types (or \emph{constraint language})
  affect the complexity and approximability of \prob{Max CSP}.
  The PCP theorem is equivalent to the existence of a
  constraint language for which \prob{Max CSP} has a hard gap at
  location 1, i.e. it is \cc{NP}-hard to distinguish between
  satisfiable instances and instances where
  at most some constant fraction of the constraints are satisfiable.
  All constraint languages, for which the \prob{CSP} problem (i.e.,
  the problem of deciding whether all constraints can
  be satisfied) is currently known to be \cc{NP}-hard,
  have a certain algebraic property. We prove that any constraint language
 with this algebraic
  property makes \prob{Max CSP} have a hard gap at location 1 which, in
  particular, implies that such problems cannot have a PTAS unless
  \cc{P} = \cc{NP}.
  We then apply this result to \prob{Max CSP} restricted to a single constraint
  type; this class of problems contains, for instance, \prob{Max Cut} and
  \prob{Max DiCut}. Assuming
  \cc{P} $\neq$ \cc{NP}, we show that such problems do not admit
  PTAS except in some trivial cases.
  Our results hold even if the number of occurrences
  of each variable is bounded by a constant. We use these results
  to partially answer open questions and strengthen results by Engebretsen et al.~[Theor. Comput. Sci., 312 (2004), pp. 17--45], 
Feder~et al. [Discrete Math., 307 (2007), pp. 386--392], Krokhin and Larose [Proc. Principles and Practice of
  Constraint Programming (2005), pp. 388--402], and Jonsson
  and Krokhin [J. Comput. System Sci., 73 (2007), pp. 691--702].
\end{abstract}

\bigskip



\pagestyle{plain}
\thispagestyle{plain}

\newpage

\section{Introduction}
Many combinatorial optimisation problems are \cc{NP}-hard so there has
been a great interest in constructing approximation algorithms for
such problems.  For some optimisation problems, there exist powerful
approximation algorithms known as \emph{polynomial-time approximation
  schemes} (PTAS).  An optimisation problem $\Pi$ has a
PTAS $A$ if, for any fixed rational $c > 1$ and for any
instance $\I$ of $\Pi$, $A(\I, c)$ returns a $c$-approximate (i.e.,
within $c$ of optimum) solution in time polynomial in $|\I|$.
There are some
well-known \cc{NP}-hard optimisation problems that have the highly
desirable property of admitting a PTAS: examples include
\prob{Knapsack}~\cite{approx-knapsack}, \prob{Euclidean
  Tsp}~\cite{tsp-ptas}, and \prob{Independent Set} restricted to
planar graphs~\cite{Ausiello99:complexity,planar-sep}.  It is also
well-known that a large number of optimisation problems do not admit
PTAS unless some unexpected collapse of complexity classes
occurs. For instance, problems like \prob{Max
  $k$-SAT}~\cite{Arora:etal:jacm98} and \prob{Independent
  Set}~\cite{Arora:Safra:jacm98} do not admit a PTAS unless
\cc{P} = \cc{NP}. We note that if $\Pi$ is a
problem that does not admit a PTAS, then there exists a
constant $c > 1$ such that $\Pi$ cannot be approximated within $c$ in
polynomial time.

The \emph{constraint satisfaction problem}
(\prob{CSP})~\cite{handbook-cp} and its optimisation variants have
played an important role in research on approximability.  For example,
it is well known that the famous PCP theorem has an equivalent reformulation
in terms of inapproximability of some
CSP~\cite{Arora:etal:jacm98,dinur-pcp,inapprox-combopt}, and the
recent combinatorial proof of this theorem~\cite{dinur-pcp} deals
entirely with CSPs.  Other important examples include H\r{a}stad's
first optimal inapproximability results~\cite{optimalinapp} and the
work around the unique games conjecture of
Khot~\cite{near-opt-ugc,ugc,optinapp-2csp}.

We will focus on a class of optimisation problems known as the
\emph{maximum constraint satisfaction problem} (\prob{Max CSP}). The
most well-known examples in this class probably are \prob{Max $k$-SAT}
and \prob{Max Cut}.

We are now ready to formally define our problem.  Let $D$ be a
finite set. A subset $R \subseteq D^n$ is a \emph{relation} and $n$
is the \emph{arity} of $R$. Let $R_D^{(k)}$ be the set of all
$k$-ary relations on $D$ and let $R_D = \cup_{i = 1}^\infty
R_D^{(i)}$. A \emph{constraint language} is a finite subset of
$R_D$.

\begin{definition}[\prob{CSP$(\Gamma)$}]
  The constraint satisfaction problem over the constraint language
  $\Gamma$, denoted \prob{CSP$(\Gamma)$}, is defined to be the
  decision problem with instance $(V, C)$, where
  \begin{itemize}
  \item $V$ is a set of variables, and
  \item $C$ is a collection of constraints $\{C_1, \ldots, C_q\}$, in
    which each constraint $C_i$ is a pair $(R_i, \tup{s_i})$ with
    $\tup{s_i}$ a list of variables of length $n_i$, called the
    constraint scope, and $R_i \in \Gamma$ is an $n_i$-ary relation in
    $R_D$, called the constraint relation.
  \end{itemize}
  The question is whether there exists an assignment $s : V
  \rightarrow D$ which satisfies all constraints in $C$ or not. A
  constraint $(R_i, (v_{i_1}, v_{i_2}, \ldots, v_{i_{n_i}})) \in C$ is
  \emph{satisfied} by an assignment $s$ if the image of the constraint
  scope is a member of the constraint relation, i.e., if $(s(v_{i_1}),
  s(v_{i_2}), \ldots, s(v_{i_{n_i}})) \in R_i$.
\end{definition}

Many combinatorial problems are subsumed by the \prob{CSP} framework;
examples include problems in graph theory~\cite{Hell:Nesetril:GH},
combinatorial optimisation~\cite{cspapprox}, and computational
learning~\cite{Dalmau:Jeavons:tcs2003}. We refer the reader
to~\cite{complexity-lang} for an introduction to this framework.

For a constraint language $\Gamma \subseteq R_D$, the optimisation
problem \prob{Max CSP$(\Gamma)$} is defined as follows:
\begin{definition}[\prob{Max CSP$(\Gamma)$}] \label{def:maxCSP}
  \prob{Max CSP$(\Gamma)$} is defined to be the optimisation problem
  with
  \begin{description}
  \item[Instance:] An instance $(V, C)$ of \prob{CSP$(\Gamma)$}.

  \item[Solution:] An assignment $s : V \rightarrow D$ to the
    variables.

  \item[Measure:] Number of constraints in $C$ satisfied by the assignment $s$.
  \end{description}
\end{definition}

We use \emph{collections} of constraints instead of just sets of
constraints as we do not have any weights in our definition of
\prob{Max CSP}. Some of our reductions will make use of copies of
one constraint to simulate something which resembles weights. We
choose to use collections instead of weights because bounded occurrence
restrictions are easier to explain in the collection setting. Note
that we prove our hardness results in this restricted setting
without weights and with a constant bound on the number of
occurrences of each variable.

Throughout this report, \prob{Max CSP$(\Gamma)$-$k$} will denote the
problem \prob{Max CSP$(\Gamma)$} restricted to instances with the
number of occurrences of each variable is bounded by $k$.  For our
hardness results we will write that \prob{Max CSP$(\Gamma)$-$B$} is
hard (in some sense) to denote that there is a $k$ such that \prob{Max
CSP$(\Gamma)$-$k$} is hard in this sense.  If a variable occurs $t$
times in a constraint which appears $s$ times in an instance, then
this would contribute $t \cdot s$ to the number of occurrences of that
variable in the instance.

\medskip

\begin{example} \label{ex:maxkcol}
  Given a (multi)graph $G=(V,E)$, the \prob{Max $k$-Cut}
  problem, $k \geq 2$, is the problem of maximising $|E'|$, $E'\sse
  E$, such that the subgraph $G'=(V,E')$ is $k$-colourable.
  For $k=2$, this problem is known simply as \prob{Max Cut}. The problem
  \prob{Max $k$-Cut} is known to be \APX-complete for any $k$ (it is Problem GT33
  in~\cite{Ausiello99:complexity}), and so has no PTAS.
  Let $N_k$ denote the binary
  disequality relation on $\{0, 1, \ldots, k-1\}$, $k \geq 2$, that
  is, $(x,y) \in N_k \iff x \neq y$. To see that \prob{Max
    CSP$(\{N_k\})$} is precisely \prob{Max $k$-Cut}, think of vertices
  of a given graph as of variables, and apply the relation to every
  pair of variables $x,y$ such that $(x,y)$ is an edge in the graph,
  with the corresponding multiplicity.
\end{example}

\medskip

Most of the early results on the complexity and approximability of
\prob{CSP} and \prob{Max CSP} were restricted to the Boolean case,
i.e. when $D=\{0,1\}$. For instance, Schaefer~\cite{gen-sat}
characterised the complexity of \prob{CSP$(\Gamma)$} for all $\Gamma$
over the Boolean domain, the approximability of \prob{Max
CSP$(\Gamma)$} for all $\Gamma$ over the Boolean domain have also been
determined~\cite{maxcsp-boolean,boolean-csp,cspapprox}.
It has been noted that the study of non-Boolean \prob{CSP} seems to
give a better understanding (when compared with Boolean \prob{CSP})
of what makes \prob{CSP} easy or hard: it appears that many
observations made on Boolean \prob{CSP} are special cases of more
general phenomena. Recently, there has been some major progress in
the understanding of non-Boolean \prob{CSP}: Bulatov has provided a
complete complexity classification of the \prob{CSP} problem over a
three-element domain \cite{CSP-3el} and also given a classification
of constraint languages that contain all unary
relations~\cite{Bulatov:lics2003}. Corresponding results for
\prob{Max CSP} have been obtained by
Jonsson~et al.~\cite{maxCSP-three} and
Deineko~et al.~\cite{maxCSP-fixed}.

We continue this line of research by studying two aspects of
non-Boolean \prob{Max CSP}. The complexity of \prob{CSP$(\Gamma)$}
is not known for all constraint languages $\Gamma$ --- it is in fact a
major open question~\cite{csp-alg,Feder:Vardi:siamjc98}. However,
the picture is not completely unknown since the complexity of
\prob{CSP$(\Gamma)$} has been settled for many constraint
languages~\cite{CSP-3el,maltsev-simple,csp-alg,maximal-CSP,alg-struct,closure-prop}.

It has been conjectured~\cite{Feder:Vardi:siamjc98} that for all constraint
languages $\Gamma$, \prob{CSP$(\Gamma)$} is either in \cc{P} or is
\cc{NP}-complete, and the refined conjecture~\cite{csp-alg} (which we refer to as the
``\prob{CSP} Conjecture'', see Section~\ref{sec:algCSP} for details) also describes the dividing line between
the two cases. Recall that if \cc{P} $\neq$ \cc{NP}, then Ladner's
Theorem~\cite{ladner} states that there are problems of intermediate
complexity, i.e., there are problems that are not in \cc{P} and not
\cc{NP}-complete. Hence, we cannot rule out a priori if there is a
constraint language $\Gamma$ such that \prob{CSP$(\Gamma)$} is
neither in \cc{P} nor \cc{NP}-complete. If the \prob{CSP} Conjecture
is true, then the family of constraint languages which are currently
known to make \prob{CSP$(\Gamma)$} \cc{NP}-complete consists of all
constraint languages with this property.

In the first part of the report we study the family of all constraint
languages $\Gamma$ such that it is currently known that
\prob{CSP$(\Gamma)$} is \cc{NP}-hard. We prove that each constraint
language in this family makes \prob{Max CSP$(\Gamma)$} have a hard gap at
location 1.  ``Hard gap at location 1'' means that it is \cc{NP}-hard
to distinguish instances of \prob{Max CSP$(\Gamma)$} in which all
constraints are satisfiable from instances where at most an
$\epsilon$-fraction of the constraints are satisfiable (for some
constant $\epsilon$ which depends on $\Gamma$)\footnotemark[5].
\footnotetext[5]{Some authors consider the promise problem {\sc Gap-CSP}$[\epsilon,1]$
where an instance is a {\sc Max CSP} instance $(V,C)$ and the
problem is to decide between the following two possibilities:
the instance is satisfiable, or
at most $\epsilon \cdot |C|$ constraints are simultaneously satisfiable.
Obviously, if a {\sc Max CSP}$(\Gamma)$ has a hard gap at location 1, then there
exists an $\epsilon$ such that the corresponding {\sc Gap-CSP}$[\epsilon,1]$
problem is {\bf NP}-hard.}
It is immediate that having a hard gap at
location 1 excludes the existence of a PTAS for
the problem. The result is proved in Section~\ref{sec:gap1}
(Theorem~\ref{th:maxCSP-unary}) and can be stated as follows (we refer
the reader to Section~\ref{sec:gap1} for an introduction to the
algebraic terminology).

\medskip

\noindent \textbf{Result A (Hardness at gap location 1 for \prob{Max
CSP}):} Let $\Gamma$ be a core constraint language and let
$\mathcal{A}$ be the algebra associated with $\Gamma$. If
$\mathcal{A}^c$ has a factor which only has projections as term
operations, then \prob{Max CSP$(\Gamma)$} has a hard gap at location
1. The result holds even if we have a constant bound on the number of
variable occurrences.

\medskip

A similar result holds when the problem is restricted to satisfiable
instances only (Corollary~\ref{satcorollary}).
We note that for the Boolean domain and without the bounded
occurrence restriction, \textbf{Result~A} follows from a result of
Khanna~et al.~\cite[Theorem~5.14]{cspapprox}.

Interestingly, the PCP
theorem is equivalent to the fact that, for {\em some} constraint
language $\Gamma$ over some finite set $D$, \prob{Max CSP$(\Gamma)$}
has a hard gap at location
1~\cite{Arora:etal:jacm98,dinur-pcp,inapprox-combopt}. Clearly,
\prob{CSP$(\Gamma)$} cannot be polynomial time solvable in this
case. For any constraint language $\Gamma$ satisfying the condition
from \textbf{Result~A}, the problem \prob{CSP$(\Gamma)$} is known to
be \cc{NP}-complete, and it has been conjectured~\cite{csp-alg} that
\prob{CSP$(\Gamma)$} is in \cc{P} for all other (core) constraint
languages (see~Section~\ref{sec:algCSP}). Thus, \textbf{Result~A}
states that \prob{Max CSP$(\Gamma)$} has a hard gap at location 1
for {\em any} constraint language such that \prob{CSP$(\Gamma)$} is
known to be \cc{NP}-complete. Moreover, if the above mentioned
conjecture holds, then \prob{Max CSP$(\Gamma)$} has a hard gap at
location 1 whenever \prob{CSP$(\Gamma)$} is not in \cc{P}.
Another equivalent reformulation of the PCP theorem states that the
problem \prob{Max 3-SAT} has a hard gap at location
1~\cite{Arora:etal:jacm98,inapprox-combopt}, and our proof of
consists of a gap preserving reduction from this
problem.

We also show how \textbf{Result~A} can be used for partially
answering two open questions. The first one was posed by
Engebretsen~et al.~\cite{inappgroupeqns} and concerns the
approximability of a finite group problem while the second
was posed by Feder~et al.~\cite{list-hom} and concerns the
hardness of \prob{CSP$(\Gamma)$} with the restriction that each
variable occurs at most a constant number of times, under the
assumption that \prob{CSP$(\Gamma)$} is \cc{NP}-complete.

The techniques we use to prove \textbf{Result~A} are partly from
universal algebra
--- such methods have previously proved to be very useful when
studying the complexity of
\prob{CSP}~\cite{CSP-3el,maltsev-simple,csp-alg,maximal-CSP,alg-struct,closure-prop}.
However, they have not previously been used to prove hardness
results for \prob{Max CSP}.
Typically, the algebraic
combinatorial property of supermodularity and the technique of
strict implementations (see Section~\ref{sec:impl} and
Section~\ref{sec:red-tech}, respectively, for the definitions) 
have been used when proving results of this kind.
This is, for instance, the case in
\cite{maxCSP-fixed,maxCSP-three} where it is proved that
for any constraint language $\Gamma$ over
$D$ such that $\Gamma$ includes the set $C_D = \{ \{(x)\} \mid x \in
D\}$ or $D$ has at most three elements, \prob{Max CSP$(\Gamma)$} is
either solvable (to optimality) in polynomial time or else
\cc{APX}-hard (in which case it cannot have a PTAS unless \cc{P} =
\cc{NP}). 

The second aspect of \prob{Max CSP} we study is the case when the
constraint language consists of a single relation; this class of
problems contains some of the most well-studied examples of \prob{Max
CSP} such as \prob{Max Cut} and \prob{Max DiCut}. Before we state this
result we need to define some terminology. For a relation $R$ we shall
say that $R$ is \emph{$d$-valid} if $(d, \ldots, d) \in R$ for $d \in
D$ and simply \emph{valid} if $R$ is $d$-valid for some $d \in
D$. Informally, our main result on this problem is (see
Theorem~\ref{th:single} for the formal statement):

\medskip

\noindent \textbf{Result B (Approximability of single relation
\prob{Max CSP}):} Let $R$ be a relation in $R_D$. If $R$ is empty or
valid, then \prob{Max CSP$(\{R\})$} is trivial.  Otherwise, there
exists a constant $c$ (which depends on $R$) such that it is
\cc{NP}-hard to approximate \prob{Max CSP$(\{R\})$} within $c$.  The
result holds even if we have a constant bound on the number of
variable occurrences.

\medskip

Jonsson and Krokhin~\cite{max-subdi} have proved that every problem
\prob{Max CSP$(\{R\})$} with $R$ neither empty nor valid is
\cc{NP}-hard. We strengthen their theorem by proving \textbf{Result~B}; to do
so, we make use of \textbf{Result~A}. Note that for some \prob{Max
CSP} problems such approximation hardness results are known, e.g.,
\prob{Max Cut} and \prob{Max DiCut}.  Our result extends those
hardness results to all possible relations.  As a corollary to this
result we get that if \prob{Max CSP$(\{R\})$} is \cc{NP}-hard, then
there is no PTAS for \prob{Max CSP$(\{R\})$} (assuming \cc{P} $\neq$
\cc{NP}). Note that a full complexity classification of
single-relation \prob{CSP} is not known. In fact, Feder and
Vardi~\cite{Feder:Vardi:siamjc98} have proved that by providing such a
classification, one has also classified the \prob{CSP} problem for
\emph{all} constraint languages.

In Section~\ref{sec:impl} we strengthen two earlier published results
on \prob{Max CSP} in various ways --- the common theme is that
\textbf{Result B} is used to obtain the results. The reader is
referred to Section~\ref{sec:impl} for definitions of the relevant
concepts used below to describe the results. We prove that unless
\cc{P} = \cc{NP}, constraint languages which contain all at most
binary 2-monotone relations on a partially ordered set which is not a
lattice order give rise to a \prob{Max CSP} problem which is hard to
approximate. The other result states that constraint languages which
contain all at most binary 2-monotone relations on a lattice and is
not supermodular on the lattice make \prob{Max CSP} hard to
approximate. These two problems have previously been studied by
Krokhin and Larose~\cite{maxcsp-diamonds-tr,maxcsp-diamonds}.

Here is an overview of the report: In Section~\ref{sec:prel} we
define some concepts we need. Section~\ref{sec:gap1} contains the
proof for our first result and Section~\ref{sec:single} contains the
proof of our second result. In Section~\ref{sec:impl} we strengthen
some earlier published results on \prob{Max CSP} as mentioned
above. We give a few concluding remarks in Section~\ref{sec:concl}.

\section{Preliminaries} \label{sec:prel}
A \emph{combinatorial optimisation problem} is defined over a set
of \emph{instances} (admissible input data); each instance $\I$
has a set $\textsf{sol}(\I)$ of \emph{feasible solutions}
associated with it, and each solution $y \in \textsf{sol}(\I)$ has
a value $m(\I, y)$. The objective is, given an instance $\I$, to
find a feasible solution of optimum value. The optimal value is
the largest one for \emph{maximisation} problems and the smallest
one for \emph{minimisation} problems.
A combinatorial optimisation problem
is said to be an $\NP$ optimisation ($\NPO$) problem if its
instances and solutions can be recognised in polynomial time, the
solutions are polynomially-bounded in the input size, and the
objective function can be computed in polynomial time (see,
e.g.,~\cite{Ausiello99:complexity}).

\begin{definition}[Performance ratio]
  A solution $s \in {\sf sol}(\I)$ to an instance $\I$ of an \cc{NPO}
  problem $\Pi$ is $r$-approximate if
\[
\max{\left\{ \frac{m(\I, s)}{\opt(\I)},\frac{\opt(\I)}{m(\I, s)} \right\} }\le r,
\]
where $\opt(\I)$ is the optimal value for a solution to $\I$. An
approximation algorithm for an \cc{NPO} problem $\Pi$ has
\emph{performance ratio} $R(n)$ if, given any instance $\I$
of $\Pi$ with $|\I|=n$, it outputs an $R(n)$-approximate
solution.
\end{definition}

\cc{PO} is the class of \cc{NPO} problems that can be solved (to
optimality) in polynomial time. An \cc{NPO} problem $\Pi$ is in the
class \cc{APX} if there is a polynomial time approximation algorithm
for $\Pi$ whose performance ratio is bounded by a constant. The
following result is well-known (see, e.g.,~\cite[Proposition~2.3]{supmod-maxCSP}).

\begin{lemma}\label{lem:in-apx}
  Let $D$ be a finite set. For every constraint language $\Gamma
  \subseteq R_D$, \prob{Max CSP$(\Gamma)$} belongs to \cc{APX}.
  Moreover, if $a$ is the maximum arity of any relation in $\Gamma$,
  then there is a polynomial time algorithm which, for every instance
  $\I = (V, C)$ of \prob{Max CSP$(\Gamma)$}, produces a solution
  satisfying at least $\frac{|C|}{|D|^a}$ constraints.
\end{lemma}

\begin{definition}[Hard to approximate] \label{def:hardtoapproximate}
  We say that a problem $\Pi$ is \emph{hard to approximate} if there
  exists a constant $c$ such that, $\Pi$ is \cc{NP}-hard to
  approximate within $c$ (that is, the existence of a polynomial-time
  approximation algorithm for $\Pi$ with performance ratio $c$ implies
  \cc{P} = \cc{NP}).
\end{definition}

The following notion has been defined in a more general setting by
Petrank~\cite{gap-location}.
\begin{definition}[Hard gap at location $\alpha$]
  \prob{Max CSP$(\Gamma)$} has a \emph{hard gap at location $\alpha
  \leq 1$} if there exists a constant $\epsilon < \alpha$ and a
  polynomial-time reduction from an \cc{NP}-complete problem $\Pi$ to
  \prob{Max CSP$(\Gamma)$} such that,
  \begin{itemize}
    \item \textsc{Yes} instances of $\Pi$ are mapped to instances $\I
  = (V, C)$ such that $\opt(\I) \geq
  \alpha |C|$, and
  \item \textsc{No} instances of $\Pi$ are mapped to instances $\I
  = (V, C)$ such that $\opt(\I) \leq \epsilon |C|$.
  \end{itemize}
\end{definition}

Note that if a problem $\Pi$ has a hard gap at location $\alpha$ (for
any $\alpha$) then $\Pi$ is hard to approximate. This simple
observation has been used to prove inapproximability results for a
large number of optimisation problems. See,
e.g.,~\cite{hardness-of-approx,Ausiello99:complexity,inapprox-combopt}
for surveys on inapproximability results and the related PCP theory.

\subsection{Approximation Preserving Reductions}
To prove our approximation hardness results we use
\emph{$AP$-reductions}. This type of reduction is most commonly used
to define completeness for certain classes of optimisation problems
(i.e., \cc{APX}). However, no \cc{APX}-hardness results are actually
proven in this report since we concentrate on proving that problems
are hard to approximate (in the sense of
Definition~\ref{def:hardtoapproximate}). We will frequently use
$AP$-reductions anyway and this is justified by
Lemma~\ref{lem:ap-red} below. Our definition of $AP$-reductions
follows~\cite{boolean-csp,cspapprox}.

\begin{definition}[$AP$-reduction]
  Given two \cc{NPO} problems $\Pi_1$ and $\Pi_2$ an
  \emph{$AP$-reduction} from $\Pi_1$ to $\Pi_2$ is a triple $(F, G,
  \alpha)$ such that,
\begin{itemize}
\item $F$ and $G$ are polynomial-time computable functions and
$\alpha$ is a constant;

\item for any instance $\I$ of $\Pi_1$, $F(\I)$ is an
instance of $\Pi_2$;

\item for any instance $\I$ of $\Pi_1$, and any feasible
solution $s'$ of $F(\I)$, $G(\I,s')$ is a feasible solution of
$\I$;

\item for any instance $\I$ of $\Pi_1$, and any $r\ge 1$, if
$s'$ is an $r$-approximate solution of $F(\I)$ then $G(\I,s')$ is
an $(1+(r-1)\alpha+o(1))$-approximate solution of $\I$ where the
$o$-notation is with respect to $|\I|$.
\end{itemize}

If such a triple exist we say that $\Pi_1$ is $AP$-reducible to
$\Pi_2$. We use the notation $\Pi_1 \leq_{AP} \Pi_2$ to denote this
fact.
\end{definition}

It is a well-known fact (see, e.g., Section~8.2.1
in~\cite{Ausiello99:complexity}) that $AP$-reductions compose. The
following simple lemma makes $AP$-reductions useful to us.
\begin{lemma} \label{lem:ap-red}
If $\Pi_1$ $\leq_{AP}$ $\Pi_2$ and $\Pi_1$ is hard to approximate,
then $\Pi_2$ is hard to approximate.
\end{lemma}
\begin{proof}
  Let $c > 1$ be the constant such that it is \cc{NP}-hard to
  approximate $\Pi_1$ within $c$. Let $(F, G, \alpha)$ be the
  $AP$-reduction which reduces $\Pi_1$ to $\Pi_2$. We will prove that
  it is \cc{NP}-hard to approximate $\Pi_2$ within
  \[
  r = \frac{1}{\alpha} (c - 1) + 1 - \epsilon'
  \]
  for any $\epsilon' > 0$.

  Let $\I_1$ be an instance of $\Pi_1$. Then, $\I_2 = F(\I_1)$ is an
  instance of $\Pi_2$. Given an $r$-approximate solution to $\I_2$ we
  can construct an $(1 + (r-1)\alpha + o(1))$-approximate solution to
  $\I_1$ using $G$. Hence, we get an $1+(r-1)\alpha + o(1) = c -
  \alpha \epsilon' + o(1)$ approximate solution to $\I_1$, and when
  the instances are large enough this is strictly smaller than $c$. As
  $c > 1$ we can choose $\epsilon'$ such that $\epsilon' > 0$ and $c -
  \alpha \epsilon' > 1$.
\end{proof}



\subsection{Reduction Techniques} \label{sec:red-tech}
The basic reduction technique in our approximation hardness proofs is
based on \emph{strict implementations} and \emph{perfect
  implementations}. Those techniques have been used before when
studying \prob{Max CSP} and other \prob{CSP}-related
problems~\cite{boolean-csp,maxCSP-three,cspapprox}.

\begin{definition}[Implementation]
  A collection of constraints $C_1, \ldots, C_m$ over a tuple of
  variables $\tup{x} = (x_1, \ldots, x_p)$ called \emph{primary
    variables} and $\tup{y} = (y_1, \ldots, y_q)$ called
  \emph{auxiliary variables} is an \emph{$\alpha$-implementation} of
  the $p$-ary relation $R$ for a positive integer $\alpha$ if the
  following conditions are satisfied:
  \begin{enumerate}
  \item For any assignment to $\tup{x}$ and $\tup{y}$, at most $\alpha$
    constraints from $C_1, \ldots, C_m$ are satisfied.
  \item For any $\tup{x}$ such that $\tup{x} \in R$, there exists an
    assignment to $\tup{y}$ such that exactly $\alpha$ constraints are
    satisfied.
  \item For any $\tup{x}, \tup{y}$ such that $\tup{x} \not \in R$, at
    most $(\alpha - 1)$ constraints are satisfied.
  \end{enumerate}
\end{definition}

\begin{definition}[Strict/Perfect Implementation] \label{def:impl}
  An $\alpha$-implementation is a \emph{strict implementation} if for
  every $\tup{x}$ such that $\tup{x} \not \in R$ there exists
  $\tup{y}$ such that exactly $(\alpha - 1)$ constraints are
  satisfied. An $\alpha$-implementation (not necessarily strict) is a
  \emph{perfect implementation} if $\alpha = m$.
\end{definition}

It will sometimes be convenient for us to view relations as predicates
instead. In this case an $n$-ary relation $R$ over the domain $D$ is a
function $r : D^n \rightarrow \{0,1\}$ such that $r(\tup{x}) = 1 \iff
\tup{x} \in R$. Most of the time we will use predicates when we are
dealing with strict implementations and relations when we are working
with perfect implementations, because perfect implementations are
naturally written as a conjunction of constraints whereas strict
implementations may naturally be seen as a sum of predicates. We will
write strict $\alpha$-implementations in the following form
\[
g(\tup{x}) + (\alpha - 1) = \max_{\tup{y}}{\sum_{i=1}^{m}{g_i(\tup{x}_i)}}
\]
where $\tup{x} = (x_1, \ldots, x_p)$ are the primary variables,
$\tup{y} = (y_1, \ldots, y_q)$ are the auxiliary variables,
$g(\tup{x})$ is the predicate which is implemented, and each
$\tup{x}_i$ is a tuple of variables from $\tup{x}$ and $\tup{y}$.

We say that a collection of relations $\Gamma$ \emph{strictly
  (perfectly) implements} a relation $R$ if, for some $\alpha \in \Z$,
there exists a strict (perfect) $\alpha$-implementation of $R$ using
relations only from $\Gamma$. It is not difficult to show that if
$R$ can be obtained from $\Gamma$ by a series of strict (perfect)
implementations, then it can also be obtained by a single strict
(perfect) implementation (for the Boolean case, this is shown
in~\cite[Lemma~5.8]{boolean-csp}).

The following lemma indicates the importance of strict
implementations for \prob{Max CSP}. It was first proved for the
Boolean case, but without the assumption on bounded occurrences,
in~\cite[Lemma~5.17]{boolean-csp}. A proof of this lemma in our
setting can be found in~\cite[Lemma~3.4]{maxCSP-fixed} (the lemma is
stated in a slightly different form but the proof establishes the
required $AP$-reduction).

\begin{lemma} \label{lem:strict}
  If $\Gamma$ strictly implements a predicate $f$, then, for any
  integer $k$, there is an integer $k'$ such that \prob{Max
    CSP$(\Gamma \cup \{f\})$-$k$} $\leq_{AP}$ \prob{Max
    CSP$(\Gamma)$-$k'$}. 
\end{lemma}

Lemma~\ref{lem:strict} will be used as follows in our proofs of
approximation hardness: if $\Gamma'$ is a fixed finite collection of
predicates each of which can be strictly implemented by $\Gamma$, then
we can assume that $\Gamma' \subseteq \Gamma$. For example, if
$\Gamma$ contains a binary predicate $f$, then we can assume, at any
time when it is convenient, that $\Gamma$ also contains
$f'(x,y)=f(y,x)$, since this equality is a strict 1-implementation of
$f'$.

For proving hardness at gap location 1, we have the following lemma.
\begin{lemma} \label{lem:gap-impl}
  If a finite constraint language $\Gamma$ perfectly implements a
  relation $R$ and \prob{Max CSP$(\Gamma \cup \{R\})$-$k$} has a hard
  gap at location 1, then \prob{Max CSP$(\Gamma)$-$k'$} has a hard gap
  at location 1 for some integer $k'$.
\end{lemma}
\begin{proof}
  Let $N$ be the minimum number of relations that are needed in a
  perfect implementation of $R$ using relations from $\Gamma$.

  Given an instance $\I = (V, C)$ of \prob{Max CSP$(\Gamma \cup
    \{R\})$-$k$}, we construct an instance $\I' = (V', C')$ of
  \prob{Max CSP$(\Gamma)$-$k'$} (where $k'$ will be specified below)
  as follows: we use the set $V''$ to store auxiliary variables during
  the reduction so we initially let $V''$ be the empty set. For a
  constraint $c = (Q, \tup{s}) \in C$, there are two cases to
  consider:
  \begin{enumerate}
  \item If $Q \neq R$, then add $N$ copies of $c$ to $C'$.
  \item If $Q = R$, then add the implementation of $R$ to $C'$ where
    any auxiliary variables in the implementation are replaced with
    fresh variables which are added to $V''$.
  \end{enumerate}
  Finally, let $V' = V \cup V''$.  It is clear that there exists an
  integer $k'$, independent of $\I$, such that $\I'$ is an instance of
  \prob{Max CSP$(\Gamma')$-$k'$}.

If all constraints are simultaneously satisfiable in $\I$, then all
constraints in $\I'$ are also simultaneously satisfiable.  On the
other hand, if $\opt(\I) \leq \epsilon |C|$ then
\begin{align}
\opt(\I') &\leq \epsilon N |C| + (1 - \epsilon) (N - 1) |C|           \notag \\
          &=    \left( \epsilon + (1 - \epsilon) (1 - 1/N) \right) |C'| . \notag
\end{align}

The inequality holds because each constraint in $\I$ introduces a
group of $N$ constraints in $\I'$ and, as $\opt(\I) \leq \epsilon |C|$,
at most $\epsilon |C|$ such groups are completely satisfied. In all
other groups (there are $(1 - \epsilon)|C|$ such groups) at least one
constraint is not satisfied. We conclude that \prob{Max
  CSP$(\Gamma)$-$k'$} has a hard gap at location 1.
\end{proof}

An important concept is that of a \emph{core}. To define cores
formally we need retractions. A \emph{retraction} of a constraint
language $\Gamma \subseteq R_D$ is a function $\pi : D \rightarrow D$
such that if $D'$ is the image of $\pi$ then $\pi(x) = x$ for all $x
\in D'$, furthermore for every $R \in \Gamma$ we have $(\pi(t_1),
\ldots, \pi(t_n)) \in R$ for all $(t_1, \ldots, t_n) \in R$. We will
say that $\Gamma$ is a \emph{core} if the only retraction of $\Gamma$
is the identity function.  Given a relation $R \in R_D^{(k)}$ and a
subset $X$ of $D$ we define the \emph{restriction of $R$ onto $X$} as
follows: $\proj{R}{X} = \{ \tup{x} \in X^k \mid \tup{x} \in R \}$. For
a set of relations $\Gamma$ we define $\proj{\Gamma}{X} = \{
\proj{R}{X} \mid R \in \Gamma \}$. If $\pi$ is a retraction of
$\Gamma$ with image $D'$, chosen such that $|D'|$ is minimal, then a
core of $\Gamma$ is the set $\proj{\Gamma}{D'}$. For constraint
language $\Gamma, \Gamma'$ we say that \emph{$\Gamma$ retracts to
$\Gamma'$} if there is a retraction $\pi$ of $\Gamma$ such that
$\pi(\Gamma) = \Gamma'$.

The intuition here is that if $\Gamma$ is not a core, then it has a
non-injective retraction $\pi$, which implies that, for every
assignment $s$, there is another assignment $\pi s$ that satisfies all
constraints satisfied by $s$ and uses only a restricted set of values.
Consequently the problem is equivalent to a problem over this smaller
set. As in the case of graphs, all cores of $\Gamma$ are isomorphic,
so one can speak about \emph{the} core of $\Gamma$.

\medskip

\begin{example}
  Every constraint language $\Gamma$ containing all unary relations is
  a core because the only retraction of the set of unary relations is
  the identity operation.
\end{example}

\medskip

The following simple lemma connects cores with non-approximability.

\begin{lemma} \label{lem:core-gap1}
  If $\Gamma'$ is the core of $\Gamma$, then, for any $k$, \prob{Max
    CSP$(\Gamma')$-$k$} has a hard gap at location 1 if and only if
  \prob{Max CSP$(\Gamma)$-$k$} has a hard gap at location 1.
\end{lemma}
\begin{proof}
  Let $\pi$ be the retraction of $\Gamma$ such that $\Gamma' = \{
  \pi(R) \mid R \in \Gamma \}$, where $\pi(R) = \{ \pi(\tup{t}) \mid
  \tup{t} \in R\}$. Given an instance $\I = (V, C)$ of \prob{Max
  CSP$(\Gamma)$-$k$}, we construct an instance $\I' = (V, C')$ of
  \prob{Max CSP$(\Gamma')$-$k$} by replacing each constraint $(R,
  \tup{s}) \in C$ by $(\pi(R), \tup{s})$.

  From a solution $s$ to $\I'$, we construct a solution $s'$ to $\I'$
  such that $s'(x) = \pi(s(x))$. Let $(R, \tup{s}) \in C$ be a
  constraint which is satisfied by $s$. Then, there is a tuple
  $\tup{x} \in R$ such that $s(\tup{s}) = \tup{x}$ so $\pi(\tup{x})
  \in \pi(R)$ and $s'(\tup{s}) = \pi(s(\tup{s})) = \pi(\tup{x}) \in
  \pi(R)$.  Conversely, if $(\pi(R), \tup{s})$ is a constraint in
  $\I'$ which is satisfied by $s'$, then there is a tuple $\tup{x} \in
  R$ such that $s'(\tup{s}) = \pi(s(\tup{s})) = \pi(\tup{x}) \in
  \pi(R)$, and $s(\tup{s}) = \tup{x} \in R$. We conclude that $m(\I,
  s) = m(\I', s')$.

  It is not hard to see that we can do this reduction in the other way
  too, i.e., given an instance $\I' = (V', C')$ of \prob{Max
    CSP$(\Gamma')$-$k$}, we construct an instance $\I$ of \prob{Max
    CSP$(\Gamma)$-$k$} by replacing each constraint $(\pi(R), \tup{s})
  \in C'$ by $(R, \tup{s})$. By the same argument as above, this
  direction of the equivalence follows, and we conclude that the lemma
  is valid.
\end{proof}

An analogous result holds for the \prob{CSP} problem, i.e., if
$\Gamma'$ is the core of $\Gamma$, then \prob{CSP$(\Gamma)$} is in
\cc{P} (\cc{NP}-complete) if and only if \prob{CSP$(\Gamma')$} is in
\cc{P} (\cc{NP}-complete); see~\cite{alg-struct} for a proof. Cores
play an important role in Section~\ref{sec:single}, too. We have the
following lemma:
\begin{lemma}[Lemma 2.11 in~\cite{maxCSP-three}] \label{lem:core}
  If $\Gamma'$ is the core of $\Gamma$, then \prob{Max
    CSP$(\Gamma')$-$B$} $\leq_{AP}$ \prob{Max CSP$(\Gamma)$-$B$}.
\end{lemma}

The lemma is stated in a slightly different form
in~\cite{maxCSP-three} but the proof establishes the required
$AP$-reduction.

\section{Result A: Hardness at Gap Location 1 for \prob{Max CSP}} \label{sec:gap1}
In this section, we will prove \textbf{Result A} which we state as
Theorem~\ref{th:maxCSP-unary}.  The proof makes use of some concepts
from universal algebra and we present the relevant definitions and
results in Section~\ref{sec:algdef} and Section~\ref{sec:algCSP}. The
proof is contained in Section~\ref{sec:unary}.

\subsection{Definitions and Results from Universal Algebra} \label{sec:algdef}
We will now present the definitions and basic results we need from
universal algebra. For a more thorough treatment of universal
algebra in general we refer the reader
to~\cite{course-univalg,univalg}. The article~\cite{csp-alg} contains a
presentation of the relationship between universal algebra and
constraint satisfaction problems.

An \emph{operation} on a finite set $D$ is an arbitrary
function $f : D^k \rightarrow D$.  Any operation on $D$ can be
extended in a standard way to an operation on tuples over $D$, as
follows: let $f$ be a $k$-ary operation on $D$. For any collection of
$k$ $n$-tuples, $\tup{t_1},\tup{t_2}, \dots, \tup{t_k} \in D^n$, the
$n$-tuple $f(\tup{t_1},\tup{t_2}, \dots ,\tup{t_k})$ is defined as
follows:
\begin{align}
f(\tup{t_1},\tup{t_2}, \ldots, \tup{t_k}) = (&f(\tup{t_1}[1],\tup{t_2}[1], \ldots, \tup{t_k}[1]),
                                              f(\tup{t_1}[2],\tup{t_2}[2], \ldots, \tup{t_k}[2]), \ldots, \notag \\
                                             &f(\tup{t_1}[n],\tup{t_2}[n], \ldots, \tup{t_k}[n])),        \notag
\end{align}
where $\tup{t_j}[i]$ is the $i$-th component in tuple $\tup{t_j}$. If
$f(d, d, \ldots, d) = d$ for all $d \in D$, then $f$ is said to be
\emph{idempotent}. An operation $f : D^k \rightarrow D$ which
satisfies $f(x_1, x_2, \ldots, x_k) = x_i$, for some $i$, is called a
\emph{projection}.

Let $R$ be a relation in the constraint language $\Gamma$. If $f$ is
an operation such that for all $\tup{t_1},\tup{t_2}, \dots, \tup{t_k}
\in R$ we have $f(\tup{t_1}, \tup{t_2}, \dots ,\tup{t_k}) \in R$,
then $R$ is said to be \emph{invariant} (or, in other words, closed) under $f$.
If all constraint relations in $\Gamma$ are invariant under $f$, then
$\Gamma$ is said to be invariant under $f$. An operation $f$ such that $\Gamma$
is invariant under $f$ is called a \emph{polymorphism} of $\Gamma$.
The set of all polymorphisms of $\Gamma$ is denoted $\Pol(\Gamma)$.
Given a set of operations $F$, the set of all relations that is
invariant under all the operations in $F$ is denoted $\Inv(F)$.

\medskip

\begin{example}
  Let $D = \{0,1,2\}$ and let $R$ be the directed cycle on $D$, i.e.,
  $R = \{(0,1),$ $(1, 2),$ $(2, 0)\}$. One polymorphism of $R$ is the
  operation $f : \{0,1,2\}^3 \rightarrow \{0,1,2\}$ defined as $f(x,
  y, z) = x - y + z \pmod{3}$. This can be verified by considering all
  possible combinations of three tuples from $R$ and evaluating $f$
  component-wise.  Let $K$ be the complete graph on $D$. It is well
  known and not hard to check that if we view $K$ as a binary relation,
  then all idempotent polymorphisms of $K$ are projections.
\end{example}

\medskip

We continue by defining a closure operator $\langle \cdot \rangle$ on
sets of relations: for any set $\Gamma \subseteq R_D$, the set
$\langle \Gamma \rangle$ consists of all relations that can be
expressed using relations from $\Gamma \cup \{ EQ_D \}$ (where $EQ_D$
denotes the equality relation on $D$), conjunction, and existential
quantification. Those are the relations definable by \emph{primitive
  positive formulae} (pp-formulae). As an example of a pp-formula
consider the relations $A = \{(0,0), (0,1), (1,0)\}$ and $B =
\{(1,0), (0,1), (1,1)\}$ over the Boolean domain $\{0,1\}$. With
those two relations we can construct $I = \{(0,0), (0,1), (1,1)\}$
with the pp-formula
\[
I(x, y) \iff \exists z: A(x, z) \land B(z, y) .
\]
Note that pp-formulae and perfect implementations from
Definition~\ref{def:impl} are the same concept.  Intuitively,
constraints using relations from $\langle \Gamma \rangle$ are
exactly those which can be simulated by constraints using relations
from $\Gamma$ in the \prob{CSP} problem. Hence, for any finite
subset $\Gamma'$ of $\langle \Gamma \rangle$, \prob{CSP$(\Gamma')$}
is not harder than \prob{CSP$(\Gamma)$}. That is, if
\prob{CSP$(\Gamma')$} is \cc{NP}-complete for some finite subset
$\Gamma'$ of $\langle \Gamma \rangle$, then \prob{CSP$(\Gamma)$} is
\cc{NP}-complete. If \prob{CSP$(\Gamma)$} is in \cc{P}, then
\prob{CSP$(\Gamma')$} is in \cc{P} for every finite subset $\Gamma'$
of $\langle \Gamma \rangle$.  We refer the reader
to~\cite{closure-prop} for a further discussion on this topic.

  The sets of relations of the form $\langle \Gamma \rangle$ are
  referred to as \emph{relational clones}, or \emph{co-clones}.  An
  alternative characterisation of relational clones is given in the
  following theorem.
\begin{theorem}[\cite{PK79}] \label{th:polinv} \
\begin{itemize}
\item For every set $\Gamma \subseteq R_D$, $\langle \Gamma \rangle =
\Inv(\Pol(\Gamma))$.

\item If $\Gamma' \subseteq \langle \Gamma \rangle$, then
$\Pol(\Gamma) \subseteq \Pol(\Gamma')$.
\end{itemize}
\end{theorem}

We will now define finite algebras and some related notions which we
need later on. The three definitions below closely follow the
presentation in~\cite{csp-alg}.

\begin{definition}[Finite algebra]
  A \emph{finite algebra} is a pair $\mathcal{A} = (A; F)$ where $A$
  is a finite non-empty set and $F$ is a set of finitary operations on
  $A$.
\end{definition}

We will only make use of finite algebras so we will write
\emph{algebra} instead of \emph{finite algebra}. An algebra is said
to be \emph{non-trivial} if it has more than one element.

\begin{definition}[Homomorphism of algebras]
  Given two algebras $\mathcal{A} = (A; F_A)$ and $\mathcal{B} = (B;
  F_B)$ such that $F_A = \{f^A_i \mid i \in I\}$, $F_B = \{f^B_i \mid
  i \in I \}$ and both $f^A_i$ and $f^B_i$ are $n_i$-ary for all $i
  \in I$, then $\phi : A \rightarrow B$ is said to be an
  \emph{homomorphism} from $\mathcal{A}$ to $\mathcal{B}$ if
  \[
  \phi(f^A_i(a_1, a_2, \ldots, a_{n_i})) = f^B_i(\phi(a_1), \phi(a_2), \ldots, \phi(a_{n_i}))
  \]
  for all $i \in I$ and $a_1, a_2, \ldots, a_{n_i} \in A$. If $\phi$
  is surjective, then $\mathcal{B}$ is a \emph{homomorphic image} of
  $\mathcal{A}$.
\end{definition}

Given a homomorphism $\phi$ mapping $\mathcal{A} = (A; F_A)$ to
$\mathcal{B} = (B; F_B)$, we can construct a equivalence relation
$\theta$ on $A$ as $\theta = \{(x, y) \mid \phi(x) = \phi(y) \}$.  The
relation $\theta$ is said to be a \emph{congruence relation} of
$\mathcal{A}$. We can now construct the \emph{quotient algebra}
$\mathcal{A}/\theta = (A/\theta; F_A/\theta)$. Here, $A/\theta = \{
x/\theta \mid x \in A \}$ and $x/\theta$ is the equivalence class
containing $x$. Furthermore, $F_A/\theta = \{ f/\theta \mid f \in F_A
\}$ and $f/\theta$ is defined such that $f/\theta(x_1/\theta,
x_2/\theta, \ldots, x_n/\theta) = f(x_1, x_2, \ldots, x_n)/\theta$.

For an operation $f : D^n \rightarrow D$ and a subset $X \subseteq D$
we define $\proj{f}{X}$ as the function $g : X^n \rightarrow D$ such
that $g(\tup{x}) = f(\tup{x})$ for all $\tup{x} \in X^n$. For a set of
operations $F$ on $D$ we define $\proj{F}{X} = \{ \proj{f}{X} \mid f
\in F \}$.

\begin{definition}[Subalgebra]
  Let $\mathcal{A} = (A; F_A)$ be an algebra and $B \subseteq A$. If
  for each $f \in F_A$ and any $b_1, b_2, \ldots, b_n \in B$, we have
  $f(b_1, b_2, \ldots, b_n) \in B$, then $\mathcal{B} = (B;
  \proj{F_A}{B})$ is a \emph{subalgebra} of $\mathcal{A}$.
\end{definition}

The operations in $\Pol(\Inv(F_A))$ are the \emph{term operations}
of $\mathcal{A}$. If all term operations are surjective, then the
algebra is said to be \emph{surjective}. Note that $\Inv(F_A)$ is a
core if and only if $\mathcal{A}$ is
surjective~\cite{csp-alg,alg-struct}.  If $F$ consist of all the
idempotent term operations of $\mathcal{A}$, then the algebra $(A;
F)$ is called the \emph{full idempotent reduct of $\mathcal{A}$},
and we will denote this algebra by $\mathcal{A}^c$. Given a set of
relations $\Gamma$ over the domain $D$ we say that the algebra
$\mathcal{A}_\Gamma=(D; \Pol(\Gamma))$ is \emph{associated} with
$\Gamma$. An algebra $\mathcal{B}$ is said to be a \emph{factor} of
the algebra $\mathcal{A}$ if $\mathcal{B}$ is a homomorphic image of
a subalgebra of $\mathcal{A}$. A \emph{non-trivial factor} is an
algebra which is not trivial, i.e., it has at least two elements.

\subsection{Constraint Satisfaction and Algebra} \label{sec:algCSP}
We continue by describing some connections between constraint
satisfaction problems and universal algebra. We will also formally
state \textbf{Result~A} in Theorem~\ref{th:maxCSP-unary}.
The following theorem concerns the hardness of \prob{CSP} for
certain constraint languages.
\begin{theorem}[\cite{csp-alg}] \label{th:CSP-unary}
  Let $\Gamma$ be a core constraint language. If
  $\mathcal{A}^c_\Gamma$ has a non-trivial factor whose term
  operations are only projections, then \prob{CSP$(\Gamma)$} is
  \cc{NP}-hard.
\end{theorem}

It has been conjectured~\cite{csp-alg}
that, for all other core languages $\Gamma$, the problem
\prob{CSP$(\Gamma)$} is tractable, and this conjecture has been
verified in many important cases (see,
e.g.,~\cite{Bulatov:lics2003,CSP-3el}).

The first main result of this report is the following theorem which
states that \prob{Max CSP$(\Gamma)$-$B$} has a hard gap at location
1 whenever the condition which makes \prob{CSP$(\Gamma)$} hard in
Theorem~\ref{th:CSP-unary} is satisfied.
\begin{theorem} \label{th:maxCSP-unary}
  Let $\Gamma$ be a core constraint language. If
  $\mathcal{A}^c_\Gamma$ has a non-trivial factor whose term
  operations are only projections, then \prob{Max CSP$(\Gamma)$-$B$}
  has a hard gap at location 1.
\end{theorem}

The proof of this result can be found in Section~\ref{sec:unary}.
Note that if the above conjecture is true then
Theorem~\ref{th:maxCSP-unary} describes all constraint languages
$\Gamma$ for which \prob{Max CSP$(\Gamma)$} has a hard gap at location
1 because, obviously, $\Gamma$ cannot have this property when
\prob{CSP$(\Gamma)$} is tractable.

There is another characterisation of the algebras in
Theorem~\ref{th:CSP-unary} which corresponds to tractable constraint
languages. To state the characterisation we need the following
definition.
\begin{definition}[Weak Near-Unanimity Function]
  An operation $f : D^n \rightarrow D$, where $n \geq 2$, is a \emph{weak near-unanimity
    function} if $f$ is idempotent and
  \[
  f(x, y, y, \ldots, y) = f(y, x, y, y, \ldots, y) = \ldots = f(y, \ldots, y, x)
  \]
  for all $x, y \in D$.
\end{definition}

Hereafter we will use the acronym \emph{wnuf} for weak
near-unanimity functions. We say that an algebra ${\mathcal A}$
{\em admits a wnuf} if there is a wnuf among the term operations
of ${\mathcal A}$. We also say that a constraint language $\Gamma$
admits a wnuf if there is a wnuf among the polymorphisms of $\Gamma$.
By combining a theorem proven by Mar\'oti and
McKenzie~\cite[Theorem~1.1]{wnuf} with a result by Bulatov and
Jeavons~\cite[Proposition~4.14]{algstruct-comb}, we get the following:
\begin{theorem} \label{th:wnuf}
  Let $\mathcal{A}$ be an idempotent algebra. The following are
  equivalent:
  \begin{itemize}
  \item There is a non-trivial factor $\mathcal{B}$ of $\mathcal{A}$
    such that $\mathcal{B}$ only have projections as term operations.
  \item The algebra $\mathcal{A}$ does not admit any wnuf.
  \end{itemize}
\end{theorem}

\subsection{Proof of \textbf{Result A}} \label{sec:unary}
We will now prove Theorem~\ref{th:maxCSP-unary}.
Let $3SAT_0$ denote the relation $\{0,1\}^3 \setminus \{(0,0,0)\}$. We
also introduce three slight variations of $3SAT_0$, let $3SAT_1 =
\{0,1\}^3 \setminus \{(1,0,0)\}$, $3SAT_2 = \{0,1\}^3 \setminus
\{(1,1,0)\}$, and $3SAT_3 = \{0,1\}^3 \setminus \{(1,1,1)\}$. To
simplify the notation we let $\Gamma_{3SAT} = \{3SAT_0,$ $3SAT_1,$
$3SAT_2,$ $3SAT_3\}$. It is not hard to see that the problem \prob{Max
CSP$(\Gamma_{3SAT})$} is precisely \prob{Max 3Sat}. It is
well-known that this problem, even when restricted to instances in
which each variable occurs at most a constant number of times, has a
hard gap at location 1, see e.g.,~\cite[Theorem~7]{inapprox-combopt}. We state
this as a lemma.

\begin{lemma}[\cite{inapprox-combopt}] \label{lem:3sat-B}
  \prob{Max CSP$(\Gamma_{3SAT})$-$B$} has a hard gap at location 1.
\end{lemma}

To prove Theorem~\ref{th:maxCSP-unary} we will utilise \emph{expander
graphs}.
\begin{definition}[Expander graph]
  A $d$-regular graph $G$ is an \emph{expander graph} if, for
  any $S \subseteq V[G]$, the number of edges between $S$ and $V[G]
  \setminus S$ is at least $\min(|S|,|V[G] \setminus S|)$.
\end{definition}

Expander graphs are frequently used for proving properties of
\prob{Max CSP}, cf.~\cite{Crescenzi:ccc97,maxsnp-j}. Typically, they
are used for bounding the number of variable occurrences.  A
concrete construction of expander graphs has been provided by
Lubotzky~et al.~\cite{ramanujan}.

\begin{theorem} \label{th:14expander}
  A polynomial-time algorithm $T$ and a fixed integer $N$
  exist such that, for any $k>N$, $T(k)$ produces a
  14-regular expander graph with $k(1+o(1))$ vertices.
\end{theorem}

There are four basic ingredients in the proof of
Theorem~\ref{th:maxCSP-unary}. The first three are
Lemma~\ref{lem:gap-impl}, Lemma~\ref{lem:3sat-B}, and the use of
expander graphs to bound the number of variable occurrences. We also
use an alternative characterisation (Lemma~\ref{lem:altchar}) of constraint
languages satisfying the conditions of the theorem. This is a slight
modification of a part of the proof of Proposition 7.9
in~\cite{csp-alg}. The implication below is in fact an equivalence and we
refer the reader to~\cite{csp-alg} for the details.
Given a function $f: D \rightarrow D$, and a relation $R \in R_D$, the
\emph{full preimage} of $R$ under $f$, denoted by $f^{-1}(R)$, is the
relation $\{ \tup{x} \mid f(\tup{x}) \in R \}$ (as usual, $f(\tup{x})$
denotes that $f$ should be applied componentwise to $\tup{x}$).

\begin{lemma} \label{lem:altchar}
  Let $\Gamma$ be a core constraint language. If the algebra
    $\mathcal{A}^c_\Gamma$ has a non-trivial factor whose term
    operations are only projections, then there is a subset $B$ of $D$
    and a surjective mapping $\phi : B \rightarrow \{0,1\}$ such that
    the relational clone $\langle \Gamma \cup C_D \rangle$ contains
    the relations $\phi^{-1}(3SAT_0)$, $\phi^{-1}(3SAT_1)$,
    $\phi^{-1}(3SAT_2)$, and $\phi^{-1}(3SAT_3)\}$.
\end{lemma}
\begin{proof}
Let $\mathcal{A}'$ be the subalgebra of $\mathcal{A}$ such that there
is a homomorphism $\phi$ from $\mathcal{A'}$ to an algebra
$\mathcal{B}$ whose term operations are only projections. We can
assume, without loss of generality, that the set $\{0,1\}$ is
contained in the universe of $\mathcal{B}$. It is easy to see that any
relation is invariant under any projections. Since $\mathcal{B}$ only
have projections as term operations, the four relations $3SAT_0$,
$3SAT_1$, $3SAT_2$ and $3SAT_3$ are invariant under the term
operations of $\mathcal{B}$. It is not hard to check
(see~\cite{csp-alg}) that the full preimages of those relations under
$\phi$ are invariant under the term operations of $\mathcal{A}'$ and
therefore they are also invariant under the term operations of
$\mathcal{A}$. By Theorem~\ref{th:polinv}, this implies
$\{\phi^{-1}(3SAT_0),$ $\phi^{-1}(3SAT_1),$ $\phi^{-1}(3SAT_2),$
$\phi^{-1}(3SAT_3)\} \subseteq \langle \Gamma \cup C_D \rangle$.
\end{proof}

We are now ready to present the proof of
Theorem~\ref{th:maxCSP-unary}. Let $S$ be a permutation group on the
set $X$.  An \emph{orbit} of $S$ is a subset $\Omega$ of $X$ such that
$\Omega = \{ g(x) \mid g \in S \}$ for some $x \in X$.

\begin{proof}
  Let $\mathcal{A}_\Gamma = (D; \Pol(\Gamma))$ be the algebra
  associated with $\Gamma$.  For any $a \in D$, we denote the unary
  constant relation containing only $a$ by $c_a$, i.e., $c_a =
  \{(a)\}$. Let $C_D$ denote the set of all constant relations over
  $D$, that is, $C_D = \{c_a \mid a \in D\}$. By Lemma~\ref{lem:altchar},
  there exists a subset (in fact, a subalgebra) $B$ of $D$ and a
  surjective mapping $\phi: B \rightarrow \{0, 1\}$ such that the
  relational clone $\langle \Gamma \cup C_D \rangle$ contains
  $\phi^{-1}(\Gamma_{3SAT}) = \{\phi^{-1}(R) \mid R \in \Gamma_{3SAT}\}$. By
  Lemma~\ref{lem:3sat-B}, \prob{Max CSP$(\Gamma_{3SAT})$-$B$} is hard
  at gap location 1, so, by Lemma~\ref{lem:core-gap1}, \prob{Max
  CSP$(\phi^{-1}(\Gamma_{3SAT}))$-$B$} is also hard at gap location 1
  (because $\Gamma_{3SAT}$ is the core of $\phi^{-1}(\Gamma_{3SAT})$).

  Since $\Gamma$ is a core, its unary polymorphisms form a permutation group $S$
  on $D$. We can without loss of generality assume
  that $D = \{1, \ldots, p\}$. It is known (see Proposition 1.3 of~\cite{Szendrei:book})
  and not hard to check (using Theorem~\ref{th:polinv}) that $\Gamma$ can perfectly
  implement the following relation:  $R_S = \{(g(1),\ldots,g(p)) \mid g \in S\}$.
  Then it can also perfectly implement the relations $EQ_i$ for
  $1 \leq i \leq p$ where $EQ_i$ is the restriction of the equality
  relation on $D$ to the orbit in $S$ which contains $i$. We have
\begin{align} \notag
 EQ_i(x,y) \iff \exists z_1, \ldots, z_{i-1}, z_{i+1}, \ldots, z_{p} :
& R_S(z_1, \ldots, z_{i-1}, x, z_{i+1}, \ldots, z_{p}) \land \\ \notag
& R_S(z_1, \ldots, z_{i-1}, y, z_{i+1}, \ldots, z_{p}) . \notag
\end{align}

For $0 \leq i \leq 3$, let $R_i$ be the preimage of $3SAT_i$ under
$\phi$. Since $R_i \in \langle \Gamma \cup C_D \rangle$, we can show
that there exists a $(p+3)$-ary relation $R'_i$ in $\langle \Gamma
\rangle$ such that
\[
  R_i = \{(x,y,z) \mid (1,2,\ldots,p,x,y,z) \in R'_i\} .
\]

Indeed, since $R_i \in \langle\Gamma \cup C_D \rangle$, $R_i$ can be
defined by a pp-formula $R_i(x,y,z) \iff \exists \mathbf{t}
\psi(\mathbf{t},x,y,z)$ (here $\mathbf{t}$ denotes a tuple of
variables) where $\psi$ is a conjunction of atomic formulas involving
predicates from $\Gamma \cup C_D$ and variables from $\mathbf{t}$ and
$\{x,y,z\}$. Note that, in $\psi$, no predicate from $C_D$ is applied
to one of $\{x,y,z\}$ because these variables can take more than one
value in $R_i$. We can without loss of generality assume that every
predicate from $C_D$ appears in $\psi$ exactly once. Indeed, if such a
predicate appears more than once, then we can identify all variables to
which it is applied, and if it does not appear at all then we can add
a new variable to $\mathbf{t}$ and apply this predicate to it.  Now
assume without loss of generality that the predicate $c_i$, $1\le i\le
p$, is applied to the variable $t_i$ in $\psi$, and $\psi=\psi_1\wedge
\psi_2$ where $\psi_1=\bigwedge_{i=1}^{p}{c_i(t_i)}$ and $\psi_2$
contains only predicates from $\Gamma\setminus C_D$. Let $\mathbf{t'}$
be the list of variables obtained from $\mathbf{t}$ by removing
$t_1,\ldots,t_{p}$. It now is easy to check that that the $(p+3)$-ary
relation $R'_i$ defined by the pp-formula $\exists \mathbf{t'}
\psi_2(\mathbf{t},x,y,z)$ has the required property.

Choose $R'_i$ to be the minimal relation in $\langle \Gamma \rangle$
such that
\[
  R_i = \{(x,y,z) \mid (1,2,\ldots,p,x,y,z)\in R'_i\} .
\]
We will show that
\[
  R'_i = \{(g(1), g(2), \ldots, g(p), g(x), g(y), g(z)) \mid g\in S, (x,y,z) \in R_i \} .
\]

The set on the right-hand side of the above equality must be
contained in $R'_i$ because $R'_i$ is invariant under all operations in
$S$. On the other hand, if a tuple
$\mathbf{b}=(b_1,\ldots,b_{p},d,e,f)$ belongs to $R'_i$, then there is
a permutation $g\in S$ such that
$(b_1,\ldots,b_{p})=(g(1),\ldots,g(p))$ (otherwise, the intersection
of this relation with $R_S\times D^3\in \langle \Gamma \rangle$
would give a smaller relation with the required property). Now note that
the tuple $(1,\ldots,p,g^{-1}(d),g^{-1}(e),g^{-1}(f))$ also belongs
to $R'_i$ implying, by the choice of $R'_i$, that
$(g^{-1}(d),g^{-1}(e),g^{-1}(f))\in R_i$. This completes the proof and
the relation $R'_i$ is as described above.

To simplify the notation, let $\Gamma' = \{R'_i \mid 0 \leq i \leq 3 \} \cup \{EQ_1, \ldots,
EQ_{p}\}$. By Lemma~\ref{lem:gap-impl}, in order to prove the theorem,
it suffices to show that \prob{Max CSP$(\Gamma')$-$B$} has a hard gap at
location 1. By Lemma~\ref{lem:3sat-B}, there is an integer $K$ such
that \prob{Max CSP$(\Gamma_{3SAT})$-$K$} has a hard gap at location 1. Choose $K$ such that $K > 14$.
By Lemma~\ref{lem:core-gap1}, \prob{Max
CSP$(\phi^{-1}(\Gamma_{3SAT}))$-$K$} has the same property. We will now
$AP$-reduce \prob{Max CSP$(\phi^{-1}(\Gamma_{3SAT}))$-$K$} to \prob{Max
CSP$(\Gamma')$-$B$}. Take an arbitrary instance $\I=(V,C)$ of
\prob{Max CSP$(\phi^{-1}(\Gamma_{3SAT}))$-$K$}, and build an instance $\I'=(V',C')$ of
\prob{Max CSP$(\Gamma')$} as follows: introduce new variables $u_1,\ldots,u_{p}$,
and replace each constraint $R_i(x,y,z)$ in $\I$ by $R'_i(u_1, \ldots,
u_{p}, x, y, z)$.  Note that every variable, except the $u_i$'s, in
$\I'$ appears at most $K$ times. We will now use expander graphs to
construct an instance $\I''$ of \prob{Max CSP$(\Gamma')$} with a constant
bound on the number of occurrences for each variables.

Let $q$ be the number of constraints in $\I$ and let $q' = \max \{N, q\}$,
where $N$ is the constant in Theorem~\ref{th:14expander}. Let $G = (W, E)$ be an
expander graph (constructed in polynomial time by the algorithm $T(q')$
in Theorem~\ref{th:14expander}) such that $W = \{w_1, w_2, \ldots,
w_m\}$ and $m \geq q$. The expander graph $T(q')$ have $q'(1+o(1))$
vertices. Hence, there is a constant $\alpha$ such that $T(q')$ has at
most $\alpha q$ vertices. For each $1\le j\le p$, we introduce $m$
fresh variables $w_1^j, w_2^j, \ldots, w_m^j$ into $\I''$. For each
edge $\{w_i, w_k\} \in E$ and $1 \leq j \leq p$, introduce $p$ copies of the
constraint $EQ_j(w_i^j, w_k^j)$ into $C''$. Let $C_1, C_2, \ldots,
C_q$ be an enumeration of the constraints in $C'$. Replace $u_j$ by
$w_i^j$ in $C_i$ for all $1 \leq i \leq q$. Finally, let $C^*$ be
the union of the (modified) constraints in $C'$ and the equality
constraints in $C''$. It is clear that each variable occurs in
$\I''$ at most $Kp$ times (recall that $p=|D|$ is a constant).

Clearly, a solution $s$ to $\I$ satisfying all constraints can be
extended to a solution to $\I''$, also satisfying all constraints,
by setting $s(w_i^j)=j$ for all $1\le i\le m$ and all $1\le j\le p$.

On the other hand, if $m(\I, s) \leq \epsilon |C|$, then let $s'$ be an
optimal solution to $\I''$. We will prove that there is a constant
$\epsilon'<1$ (which depends on $\epsilon$ but not on $\I$) such that
$m(\I'', s') \leq \epsilon' |C^*|$.

We first prove that, for each $1 \leq j \leq p$, we can assume
that all variables in $W^j = \{w_1^j, w_2^j, \ldots, w_m^j\}$
have been assigned the same value by $s'$ and that all constraints in
$C''$ are satisfied by $s'$. We show that given a solution $s'$ to
$\I''$, we can construct another solution $s_2$ such that $m(\I'', s_2)
\geq m(\I'', s')$ and $s_2$ satisfies all constraints in $C''$.

Let $a^j$ be the value that at least $m/p$ of the variables in $W^j$
have been assigned by $s'$. We construct the solution $s_2$ as follows:
$s_2(w_i^j) = a^j$ for all $i$ and $j$, and $s_2(x) = s'(x)$ for all
other variables.

If there is some $j$ such that $X = \{ x \in W^j \mid s'(x) \neq
a^j\}$ is non-empty, then, since $G$ is an expander graph,  there
are at least $p\cdot\min(|X|, |W^j \setminus X|)$ constraints in
$C''$ which are not satisfied by $s'$. Note that by the choice of
$X$, we have $|W^j \setminus X| \ge m/p$ which implies
$p\cdot\min(|X|, |W^j \setminus X|)\ge |X|$.
By changing the value of the variables in $X$, we will make at most
$|X|$ non-equality constraints in $C^*$ unsatisfied because each of
the variables in $W^j$ occurs in at most one non-equality constraint
in $C^*$. In other words, when the value of the variables in $X$ are
changed we gain at least $|X|$ in the measure as some of the
equality constraints in $C''$ will become satisfied, furthermore we lose
at most $|X|$ by making at most $|X|$ constraints in $C^*$
unsatisfied. We conclude that $m(\I', s_2) \geq m(\I', s')$. Thus,
we may assume that all equality constraints in $C''$ are satisfied
by $s'$.

Since the expander graph $G$ is 14-regular and has at most $\alpha
q$ vertices, it has at most $\frac{14}{2}\alpha q$ edges. Hence, the
number of equality constraints in $C''$ is at most $7\alpha qp$, and
$|C''|/|C'|\leq 7\alpha p$. We can now bound $m(\I'', s_2)$ as follows:
\begin{align}
  m(\I'', s_2) \leq \opt(\I') + |C''| \leq
  \frac{\epsilon |C'| + |C''|}{|C'| + |C''|} (|C'| + |C''|) \leq
  \frac{\epsilon + 7\alpha p}{1+7\alpha p} (|C'| + |C''|) . \notag
\end{align}

Since $|C^*|=|C'| + |C''|$, it remains to set
$\epsilon'=\frac{\epsilon + 7\alpha p}{1+7\alpha p}$.
\end{proof}

We finish this section by using Theorem~\ref{th:maxCSP-unary} to
answer, at least partially, two open questions. The first one
concerns the complexity of \prob{CSP$(\Gamma)$-$B$}. In particular,
the following conjecture has been made by
Feder~et al.~\cite{list-hom}.

\begin{conjecture}
For any fixed $\Gamma$ such that \prob{CSP$(\Gamma)$} is
\cc{NP}-complete there is an integer $k$ such that
\prob{CSP$(\Gamma)$-$k$} is \cc{NP}-complete.
\end{conjecture}

Under the assumption that the CSP conjecture (that all problems
\prob{CSP$(\Gamma)$} not covered by Theorem~\ref{th:CSP-unary} are
tractable) holds, an affirmative answer follows immediately from
Theorem~\ref{th:maxCSP-unary}. So for all constraint languages
$\Gamma$ such that \prob{CSP$(\Gamma)$} is currently known to be
\cc{NP}-complete it is also the case that \prob{CSP$(\Gamma)$-$B$} is
\cc{NP}-complete.

The second result concerns the approximability of equations over
non-abelian groups. Petrank~\cite{gap-location} has noted that
hardness at gap location 1 implies the following: suppose that we
restrict ourselves to instances of \prob{Max CSP$(\Gamma)$} such that
there exist solutions that satisfy all constraints, i.e. we
concentrate on \emph{satisfiable} instances. Then, there exists a
constant $c$ (depending on $\Gamma$) such that no polynomial-time
algorithm can approximate this problem within $c$ (unless \cc{P} =
\cc{NP}). We get the following result for satisfiable instances:

\begin{corollary} \label{satcorollary}
  Let $\Gamma$ be a core constraint language and let $\mathcal{A}$ be
  the algebra associated with $\Gamma$. Assume there is a factor
  $\mathcal{B}$ of $\mathcal{A}^c$ such that $\mathcal{B}$ only have
  projections as term operations. Then, there exists a constant $c$
  such that \prob{Max CSP$(\Gamma)$-$B$} restricted to satisfiable
  instances cannot be approximated within $c$ in polynomial time
  (unless \cc{P} = \cc{NP}).
\end{corollary}

We will now use this observation for studying a problem
concerning groups.  Let $\mathcal{G}=(G,\cdot)$ denote a finite
group with identity element $1_G$.  An \emph{equation} over a set of
variables $V$ is an expression of the form $w_1 \cdot \ldots \cdot
w_k = 1_G$, where $w_i$ (for $1 \leq i \leq k$) is either a
variable, an inverted variable, or a group constant.
Engebretsen~et al.~\cite{inappgroupeqns} have studied the following
problem:
\begin{definition}[\prob{Eq$_\mathcal{G}$}] \label{def:eqG}
  The computational problem \prob{Eq$_\mathcal{G}$} (where
  $\mathcal{G}$ is a finite group) is defined to be the optimisation
  problem with

\smallskip

  \begin{description}
  \item[Instance:] A set of variables $V$ and a collection of
    equations $E$ over $V$.

  \item[Solution:] An assignment $s : V \rightarrow G$ to the
    variables.

  \item[Measure:] Number of equations in $E$ which are satisfied by $s$.
  \end{description}
\end{definition}

\smallskip

The problem \prob{Eq$_\mathcal{G}^1$[3]} is the same as
\prob{Eq$_\mathcal{G}$} except for the additional restrictions that
each equation contains exactly three variables and no equation
contains the same variable more than once. Their main result was the
following inapproximability result:

\begin{theorem}[Theorem~1 in~\cite{inappgroupeqns}] \label{th:inappeq}
  For any finite group $\mathcal{G}$ and constant $\epsilon > 0$, it
  is \cc{NP}-hard to approximate \prob{Eq$_\mathcal{G}^1$[3]} within
  $|G| - \epsilon$.
\end{theorem}

Engebretsen~et al.~left the approximability of
\prob{Eq$_\mathcal{G}^1$[3]} for satisfiable instances as an open
question. We will give a partial answer to the approximability of
satisfiable instances of \prob{Eq$_\mathcal{G}$}.

It is not hard to see that for any integer $k$, the equations with at
most $k$ variables over a finite group can be viewed as a constraint
language. For a group $\mathcal{G}$, we denote the constraint language
which corresponds to equations with at most three variables by
$\Gamma_\mathcal{G}$. Hence, for any finite group $\mathcal{G}$, the
problem \prob{Max CSP$(\Gamma_\mathcal{G})$} is no harder than
\prob{Eq$_\mathcal{G}$}.

Goldmann and Russell~\cite{groupeqns} have shown that
\prob{CSP$(\Gamma_\mathcal{G})$} is \cc{NP}-hard for every finite
non-abelian group $\mathcal{G}$. This result was extended to more
general algebras by Larose and Z\'adori~\cite{univeqns}.  They also
showed that for any non-abelian group $\mathcal{G}$, the algebra
$\mathcal{A} = (G; \Pol(\Gamma_\mathcal{G}))$ has a non-trivial factor
$\mathcal{B}$ such that $\mathcal{B}$ only have projections as
term operations.  We now combine Larose and Z\'adori's result with
Theorem~\ref{th:maxCSP-unary}:
\begin{corollary}
  For any finite non-abelian group $\mathcal{G}$,
  \prob{Eq$_\mathcal{G}$} has a hard gap at location 1.
\end{corollary}

Thus, there is a constant $c$ such that no polynomial-time algorithm
can approximate satisfiable instances of \prob{Eq$_\mathcal{G}$}
better than $c$, unless $\cc{P} = \cc{NP}$. There also exists a
constant $k$ (depending on the group $\mathcal{G}$) such that the
result holds for instances with variable occurrence bounded by $k$.

\section{Result B: Approximability of Single Relation \prob{Max CSP}} \label{sec:single}
In this section, we will prove the following theorem:
\begin{theorem} \label{th:single}
  Let $R \in R_D^{(n)}$ be non-empty. If $(d, \ldots, d) \in R$ for
  some $d \in D$, then \prob{Max CSP$(\{R\})$} is solvable in linear
  time. Otherwise, \prob{Max CSP$(\{R\})$-$B$} is hard to
  approximate.
\end{theorem}

The proof makes crucial use of Theorem~\ref{th:maxCSP-unary} and it
can be divided into a number of steps:
\begin{enumerate}
\item Lemma~\ref{lem:double-edge} together with
  Lemma~\ref{lem:undircycle} proves that directed cycles are
  hard to approximate (i.e., the theorem holds when $R$ is the
  edge relation of a directed cycle).

\item Vertex-transitive digraphs which are not directed cycles are
  proved to be hard to approximate in Lemma~\ref{lem:trans}.

\item Lemma~\ref{lem:bipartite} give approximation hardness for
  bipartite digraphs.

\item Lemma~\ref{lem:non-transitive} reduces the non-vertex transitive
  case to the vertex-transitive case.

\item Lemma~\ref{lem:arity} reduces general relations to binary
  relations, i.e., to digraphs. \label{it:arity}

\item Finally, Theorem~\ref{th:single} is proved by assembling the
  results from the previous sections.
\end{enumerate}

As indicated by the list above the bulk of the work deals with binary
relations.

\subsection{Approximability of Binary Relations} \label{sec:bin-rel}
In this section, we will prove that non-empty non-valid binary
relations give rise to \prob{Max CSP} problems which are hard to
approximate. Subsection~\ref{sec:trans-digraph} deals with binary
(not necessarily symmetric) relations having a transitive
automorphism group, and Section~\ref{sec:general-digraph} deals with general
binary relations.

Sometimes it will be convenient for us to view binary relations as
digraphs. A \emph{digraph} is a pair $(V, E)$ such that $V$ is a
finite set and $E \subseteq V \times V$. A \emph{graph} is a digraph
$(V, E)$ such that for every pair $(x, y) \in E$ we also have $(y, x)
\in E$. Let $R \in R_D$ be a binary relation. As $R$ is binary it can be
viewed as a digraph $G$ with vertex set $V[G] = D$ and edge set $E[G]
= R$.  We will mix freely between those two notations. For example, we
will sometimes write $(x, y) \in G$ with the intended meaning $(x,y)
\in E[G] = R$.

Let $G$ be a digraph, $R=E[G]$, and let $\Aut(G)$ denote the
automorphism group of $G$. If $\Aut(G)$ is transitive (i.e., contains
a single orbit), then we say that $G$ is
\emph{vertex-transitive}. If $D$ can be partitioned into two sets,
$A$ and $B$, such that for any $x, y \in A$ (or $x,y \in B$) we have
$(x, y) \not \in R$, then $R$ (and $G$) is \emph{bipartite}. The
\emph{directed cycle of length
  $n$} is the digraph $G$ with vertex set $V[G] = \{0, 1, \ldots,
n-1\}$ and edge set $E[G] = \{ (x, x+1) \mid x \in V[G] \}$, where the
addition is modulo $n$.  Analogously, the \emph{undirected cycle of
  length $n$} is the graph $H$ with vertex set $V[H] = \{0, 1, \ldots,
n-1\}$ and edge set $E[H] = \{(x, x+1) \mid x \in V[H] \} \cup \{(x+1,
x) \mid x \in V[H] \}$ (also in this case the additions are modulo
$n$). The undirected path with two vertices will be denoted by $P_2$.




\subsubsection{Vertex-transitive Digraphs} \label{sec:trans-digraph}
We will now tackle non-bipartite vertex-transitive digraphs and prove
that they give rise to \prob{Max CSP} problems which are hard at gap
location 1. To do this, we make use of the algebraic framework which
we used and developed in Section~\ref{sec:gap1}. Recall that we denote
the unary constant relations over a domain $D$ by $C_D$, i.e., $C_D =
\{ \{(x)\} \mid x \in D \}$. We will need certain hardness results in
the forthcoming proofs.

\begin{theorem}[\cite{hcol-rev}] \label{th:hcol-alg}
Let $G$ be an undirected core graph and let $\mathcal{A}_G$ be the
  algebra associated with $G$.  If $G$ is not bipartite, then there is
  a factor of $\mathcal{A}_G^c$ which only have projections as term
  operations.
\end{theorem}


\begin{lemma} \label{lem:v-trans-34}
  Let $G$ be a vertex-transitive core digraph such that $|V[G]| = 3$
  or $|V[G]| = 4$. If $G$ is not a directed cycle, then $G$ does not
  admit a wnuf.
\end{lemma}
\begin{proof}
 Let $v$ and $u$ be two vertices in a vertex-transitive core
  digraph. Note that the in- and out-degrees of $u$ and $v$ must
  coincide, and hence the in- and out-degrees of $v$ must be the same.

  Having this in mind, it is easy to see that there are only two
  digraphs with three vertices satisfying the conditions in the lemma:
  the directed cycle and the complete graph on three vertices.
  Similarly, it is easy to see that there are three
  core digraphs on four vertices which are vertex-transitive:
  the directed cycle, the complete graph on four vertices and
  the digraph in Figure~\ref{fig:v-trans-4}.

  For complete graphs, the results follows from
  Theorems~\ref{th:hcol-alg} and~\ref{th:wnuf}. Denote the digraph in
  Figure~\ref{fig:v-trans-4} by $G$ and consider the following perfect
  implementation (originally used by MacGillivray~\cite[step~3 in
  Theorem~3.4]{vertex-trans}).
  \[
  H(x, y) \iff \exists u, v: G(x, u) \land G(u, v) \land G(v, u) \land G(v, y)
  \]
  It is not hard to see that $H$ is the complete graph on four
  vertices. Since $H$ does not admit any wnuf, Theorem~\ref{th:polinv}
  implies that $G$ does not admit a wnuf either.
\end{proof}

\newcommand{\arr}[1]{\ar@{-}[#1] |-{\object@{>}} }
\newcommand{\uarr}[1]{\ar@{-}[#1]}

\begin{figure}
    \centering
    \[
    \xymatrix{ {\bullet} \arr{r}           & {\bullet} \arr{d} \\
               {\bullet} \arr{u} \uarr{ur} & {\bullet} \arr{l} \uarr{ul} }
    \]
    \caption{The non-trivial case in Lemma~\ref{lem:v-trans-34}} \label{fig:v-trans-4}
\end{figure}
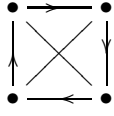

\begin{lemma} \label{lem:retract-pol}
Let $G$ and $H$ be two digraphs such that there is a retraction from
$G$ to $H$. If $G$ admits a wnuf, then so does $H$.
\end{lemma}
\begin{proof}
Let $r$ be a retraction from $G$ to $H$. If $G$ admits a wnuf
$f(x_1,\ldots,x_n)$, then it is easy to check that $H$ admits
the wnuf $r(f(x_1,\ldots,x_n))$.
\end{proof}

\begin{lemma} \label{th:trans-alg-andrei}
  If $G$ is a vertex-transitive digraph that does not retract to a
  directed cycle, then $G$ admits no wnuf.
\end{lemma}
\begin{proof}
  For the sake of contradiction, let $G$ be a digraph with the minimum
  number of vertices such that $G$ is vertex-transitive, does not
  retract to a directed cycle and $G$ admits a wnuf. Furthermore,
  among all counterexamples with $|V[G]|$ vertices, let $G$ be the one
  with the maximum number of edges.

  It is well known, and easy to show, that the core of a
  vertex-transitive digraph is also vertex-transitive.  If $G$ is not
  a core, then the core of $G$ is vertex-transitive, admits a wnuf
  (Lemma~\ref{lem:retract-pol}), and does not retract to a directed
  cycle. Hence, if $G$ is not a core, then the core of $G$ is a
  smaller counterexample. We can therefore, without loss of
  generality, assume that $G$ is a core. By
  Lemma~\ref{lem:v-trans-34}, we can assume that $|V[G]| > 4$. We need
  the latter assumption because this is assumed in the proof of
  Theorem~3.4 in~\cite{vertex-trans}, which we use below.

  For a digraph $H$, let $undir(H)$ be the digraph induced by the
  double edges of $H$. It is easy to see that $undir(H)$ can be
  perfectly implemented from $H$ as follows:
\begin{align}
undir(H)(x, y) \iff H(x, y) \land H(y, x) . \label{eq:undir}
\end{align}

  The proof of Theorem~3.4
  in~\cite{vertex-trans} shows that it is possible to perfectly
  implement a digraph $H$ with $G$ and $C_{V[G]}$ such that
  there is a retraction $r$ from $H$ to a digraph $H'$ which is
  vertex-transitive and
  \begin{enumerate}
    \item $undir(H')$ is not bipartite and not valid, or
    \item $|V[H']| < |V[G]|$ and $H'$ does not retract to a directed cycle, or
    \item $|V[H']| = |V[G]|$ and $|E[H']| > |E[G]|$ and $H'$ does not retract to a directed cycle.
  \end{enumerate}

In fact, the proof in~\cite{vertex-trans} uses constructions called
``indicator'' and ``subindicator'' to obtain $H$ from $G$, but these
constructions are well known to precisely correspond to certain
perfect implementations (or pp-formulas, see~\cite{hcol-rev} for
details).

Note that, since any wnuf is idempotent, the constraint language
$\{E[G]\}\cup C_{V[G]}$ admits a wnuf. Then, by
Theorem~\ref{th:polinv}, the digraph $H$ admits a wnuf.
Lemma~\ref{lem:retract-pol} applied to $H$ shows that $H'$ admits a
wnuf as well.  Now we see that cases (2) and (3) are impossible,
since $H'$ would contradict the choice of $G$. Case (1) leads to a
contradiction too because, by Lemmas~\ref{th:hcol-alg}
and~\ref{lem:retract-pol}, the core of $undir(H')$, which is a
non-bipartite undirected graph, would also admit a wnuf which is
impossible by Theorems~\ref{th:hcol-alg} and~\ref{th:wnuf}.
\end{proof}

\begin{corollary} \label{lem:trans}
  Let $H$ be a vertex-transitive core digraph which is not valid and
  not a directed cycle. Then, \prob{Max CSP$(\{H\})$-$B$} has a hard
  gap at location 1.
\end{corollary}
\begin{proof}
  Immediately follows from Lemma~\ref{th:trans-alg-andrei},
  Theorem~\ref{th:wnuf}, and Theorem~\ref{th:maxCSP-unary}
\end{proof}

The next lemmas help to deal with the remaining vertex-transitive
graphs, i.e. those that retract to a directed cycle.

\begin{lemma} \label{lem:undircycle}
  If $G$ is the undirected path with two vertices $P_2$, or an
  undirected cycle $C_k$, $k > 2$, then \prob{Max CSP$(\{G\})$-$B$} is
  hard to approximate.
\end{lemma}
 \begin{proof}
   If $G = P_2$, then the result follows from Example~\ref{ex:maxkcol}. If
   $G = C_k$ and $k$ is even, then the core of $C_k$ is isomorphic to
   $P_2$ and the result follows from
   Lemmas~\ref{lem:core},~\ref{lem:ap-red} combined with
   Example~\ref{ex:maxkcol}.

   From now on, assume that $G = C_k$, $k$ is odd, and $k \geq 3$. We
   will show that we can strictly implement $N_k$, i.e., the
   inequality relation. We use the following strict implementation
   \begin{align}
   N_k(z_1, z_{k-1}) + (k - 3) = \max_{z_2, z_3, \ldots, z_{k-2}} &C_k(z_1, z_2) + C_k(z_2, z_3) + \ldots +      \notag \\
                                                                  &C_k(z_{k-3}, z_{k-2}) + C_k(z_{k-2}, z_{k-1}) \notag .
   \end{align}
   It is not hard to see that if $z_1 \neq z_{k-1}$, then all $k-2$
   constraints on the right hand side can be satisfied. If $z_1 =
   z_{k-1}$, then $k-3$ constraints are satisfied by the assignment
   $z_i = z_1 + i - 1$, for all $i$ such that $1 < i < k-1$ (the
   addition and subtraction are modulo $k$). Furthermore, no
   assignment can satisfy all constraints. To see this, note that such
   an assignment would define a path $z_1, z_2, \ldots, z_{k-1}$ in
   $C_k$ with $k-2$ edges and $z_1 = z_{k-1}$. This is impossible
   since $k-2$ is odd and $k-2 < k$ .

   The lemma now follows from Lemmas~\ref{lem:strict}
   and~\ref{lem:ap-red} together with Example~\ref{ex:maxkcol}.
\end{proof}

\begin{lemma} \label{lem:double-edge}
  If $G$ is a digraph such that $(x, y) \in E[G] \Rightarrow (y, x)
  \not \in E[G]$, then \prob{Max CSP$(\{H\})$-$B$} $\leq_{AP}$
  \prob{Max CSP$(\{G\})$-$B$}, where $H$ is the undirected graph
  obtained from $G$ by replacing every edge in $G$ by two edges in
  opposing directions in $H$.
\end{lemma}
\begin{proof}
  $H(x, y) + (1 - 1) = G(x, y) + G(y, x)$ is a strict implementation
  of $H$ and the result follows from Lemma~\ref{lem:strict}.
\end{proof}

\begin{lemma} \label{lem:trans-digraph}
  If $G$ is a non-empty non-valid vertex-transitive digraph, then
  \prob{Max CSP$(\{G\})$-$B$} is hard to approximate.
\end{lemma}
\begin{proof}
  By Lemmas~\ref{lem:core} and~\ref{lem:ap-red}, it is enough to
  consider cores. For directed cycles, the results follows from
  Lemmas~\ref{lem:undircycle} and~\ref{lem:double-edge}, and, for all other
  digraphs, from Corollary~\ref{lem:trans}.
\end{proof}

\subsubsection{General Digraphs} \label{sec:general-digraph}
The main lemma of this section is Lemma~\ref{lem:digraph} which proves
our result for general digraphs. We begin by considering bipartite
digraphs.

\begin{lemma} \label{lem:bipartite}
  If $G$ is a bipartite digraph which is neither empty nor valid,
  then \prob{Max CSP$(\{G\})$-$B$} is hard to approximate.
\end{lemma}
\begin{proof}
  If there are two edges $(x, y), (y, x) \in E[G]$, then the core of
  $G$ is isomorphic to $P_2$ and the result follows from
  Lemmas~\ref{lem:ap-red} and~\ref{lem:core} together with
  Example~\ref{ex:maxkcol}. If no such pair of edges exist, then
  Lemmas~\ref{lem:ap-red} and~\ref{lem:double-edge} reduce this case to
  the previous case where there are two edges $(x, y), (y, x) \in
  E[G]$.
\end{proof}

We will use a technique known as {\em domain
restriction}~\cite{maxCSP-fixed} in the sequel. For a subset $D'
\subseteq D$, let $\proj{\Gamma}{D'} = \{ \proj{R}{D'} \mid R \in
\Gamma \mbox{\ and } \proj{R}{D'} \mbox{ is
  non-empty}\}$. The following lemma was proved
in~\cite[Lemma~3.5]{maxCSP-fixed} (the lemma is stated in a slightly
different form there, but the proof together
with~\cite[Lemma~8.2]{Ausiello99:complexity} and
Lemma~\ref{lem:in-apx} implies the existence of the required
$AP$-reduction).
\begin{lemma} \label{lem:domrest}
  Let $D'\subseteq D$ and $D' \in \Gamma$, then \prob{Max
    CSP$(\proj{\Gamma}{D'})$-$B$} $\leq_{AP}$ \prob{Max
    CSP$(\Gamma)$-$B$}.
\end{lemma}

Typically, we will let $D'$ be an orbit in the automorphism group of
a graph. We are now ready to present the
three lemmas that are the building blocks of Lemma~\ref{lem:non-transitive}.
Let $G$ be a digraph. For a set $A \subseteq V[G]$, we define $A^+ =
\{ j \mid (i, j) \in E[G], i \in A \}$, and $A^- = \{ i \mid (i, j)
\in E[G], j \in A\}$.

\begin{lemma} \label{lem:unary-union}
  If a constraint language $\Gamma$ contains two unary predicates
  $S, T$ such that $S \cap T = \emptyset$, then
  $\Gamma$ strictly implements $S \cup T$.
\end{lemma}
\begin{proof}
  Let $U = S \cup T$.  Then $U(x) + (1-1) = S(x) + T(x)$ is a strict
  implementation of $U(x)$.
\end{proof}

\begin{lemma} \label{lem:orbit+-impl}
  Let $H$ be a core digraph and $\Omega$ an orbit in $\Aut(H)$.  Then,
  $H$ strictly implements $\Omega^+$ and $\Omega^-$.
\end{lemma}
\begin{proof}
  Assume that $H \in R_D$ where $D = \{1, 2, \ldots, p\}$
  and (without loss of generality) assume that $1 \in \Omega$. We
  construct a strict implementation of $\Omega^+$; the other case can
  be proved in a similar way.  Consider the function
\[g(z_1,\ldots,z_p) = \sum_{H(i,j)=1} H(z_i,z_j).\]
By combining the fact that $H$ is a core with Theorem~1
in~\cite{expressive-power},
one sees that the following holds: $g(z_1,\ldots,z_p)=|E[H]|$ if
and only if the function $\{1 \mapsto z_1,\ldots,p \mapsto z_p\}$
is an automorphism of $H$. This also implies that
a necessary condition for $g(z_1,\ldots,z_p)=|E[H]|$ is that $z_1$
is assigned some element in the orbit containing $1$,
i.e. the orbit $\Omega$.
We claim that $\Omega^+$ can be strictly implemented
  as follows:
\[\Omega^+(x) + (\alpha - 1) = \max_{\tup{z}} \left( H(z_1, x) + g(\tup{z}) \right) \notag\]
  where $\tup{z} = (z_1, z_2, \ldots, z_p)$ and $\alpha = |E[H]|+1$.

Assume first that $x \in \Omega^+$ and choose $y \in \Omega$
such that $H(y,x)=1$. Then, there exists an automorphism
$\sigma$ such that $\sigma(1)=y$ and
$H(z_1, x) + g(\tup{z}) = 1 + |E[H]|$
by assigning variable $z_i$, $1 \leq i \leq p$, the
value $\sigma(i)$.

If $x \not\in \Omega^+$, then there is no $y \in \Omega$
such that $H(y,x)=1$. If the constraint $H(z_1,x)$ is to be
satisfied, then $z_1$ must be chosen such that $z_1 \not\in \Omega$.
We have already observed that such an assignment cannot be extended
to an automorphism of $H$ and, consequently,
$H(z_1, x) + g(\tup{z}) < 1 + |E[H]|$ whenever $z_1 \not\in \Omega$.
However, the assignment $z_i=i$, $1 \leq i \leq p$, makes
$H(z_1, x) + g(\tup{z}) = |E[H]|$
since the identity function is an
automorphism of $H$.
\end{proof}

\begin{lemma} \label{lem:orbit-rest}
  If $H$ is a core digraph and $\Omega$ an orbit in $\Aut(H)$, then,
  for every $k$, there is a number $k'$ such that \prob{Max
  CSP$(\{H|_{\Omega}\})$-$k$} $\leq_{AP}$ \prob{Max
  CSP$(\{H\})$-$k'$}.
\end{lemma}
\begin{proof}
  Let $V[H] = \{1, 2, \ldots, p\}$ and arbitrarily choose one element
  $d \in \Omega$.  Let $\I = (V, C)$ be an arbitrary instance of
  \prob{Max CSP$(\{H|_{\Omega}\})$-$k$} and let $V = \{v_1, \ldots,
  v_n\}$. We construct an instance $\I' = (V' \cup V, C' \cup C)$ of
  \prob{Max CSP$(\{H\})$-$k'$} ($k'$ will be specified below) as
  follows: for each variable $v_i \in V$:
\begin{enumerate}
\item Add fresh variables $w^1_i, \ldots, w^p_i$ to $V'$. For each
  $(a, b) \in E[H]$, add $k$ copies of the constraint $H(w^a_i,
  w^b_i)$ to $C'$.

\item Identify the variables $v_i$ and $w^d_i$ and remove $v_i$ from $V'$.
\end{enumerate}

It is clear that there exist an integer $k'$, independent of $\I'$,
such that $\I'$ is an instance of \prob{Max CSP$(\{H\})$-$k'$}.

Let $s'$ be a solution to $\I'$. For an arbitrary variable $v_i \in
V$, if there is some constraint in $C'$ which is not satisfied by
$s'$, then we can get another solution $s''$ by modifying $s'$ so that
every constraint in $C'$ is satisfied (if $H(w^a_i, w^b_i)$ is a
constraint which is not satisfied by $s'$ then set $s''(w^a_i) = a$
and $s''(w^b_i) = b$). We will denote this polynomial-time algorithm
by $P'$, so $s'' = P'(s')$. The corresponding solution to $\I$ will be
denoted by $P(s')$, so $P(s')(v_i) = P'(s')(w^d_i)$.

The algorithm $P$ may make some of the constraints involving $v_i$
unsatisfied. However, the number of copies, $k$, of the constraints in
$C'$ implies that $m(\I', s') \leq m(\I', P'(s'))$. In particular, this
means that any optimal solution to $\I'$ can be used to construct
another optimal solution which satisfies all constraints in $C'$.

Hence, for each $v_i \in V$, all constraints from step 1 are satisfied
by $s'' = P'(s')$. As $H$ is a core, $s''$ restricted to $w^1_i,
\ldots, w^p_i$ (for any $v_i \in V$) induces an automorphism of $H$.
Denote the automorphism by $f : V[H] \rightarrow V[H]$ and note that
$f$ can be defined as $f(x) = s''(w^x_i)$. Furthermore, $s''(w^d_i)
\in \Omega$ for all $w^d_i \in V$ since $d \in \Omega$.

To simplify the notation we let $l = |E[H]|$. By a straightforward
probabilistic argument we have $\opt(\I) \geq \frac{l}{p^2} |C|$.
Using this fact and the argument above we can bound the optimum of
$\I'$ as follows:
\begin{align}
\opt(\I') &\leq \opt(\I) +  kl|V|         \notag \\
          &\leq \opt(\I) + 2kl|C|         \notag \\
          &\leq \opt(\I) + 2kp^2 \opt(\I) \notag \\
          &=    (1 + 2kp^2) \opt(\I) .    \notag
\end{align}


From Lemma~\ref{lem:in-apx} we know that there exists a
polynomial-time approximation algorithm $A$ for \prob{Max
CSP$(\proj{H}{\Omega})$}.  Let us assume that $A$ is a
$c$-approximation algorithm, i.e., it produces solutions which are
$c$-approximate in polynomial time. We construct the algorithm $G$ in
the $AP$-reduction as follows:
\[
G(\I, s') =
\left\{ \begin{array}{ll}
P(s') & \textrm{if } m(\I, P(s')) \geq m(\I, A(\I)), \\
A(\I) & \textrm{otherwise.}
\end{array} \right.
\]
We see that $\opt(\I) / m(\I, G(\I, s')) \leq c$. Let $s'$ be a
$r$-approximate solution to $\I'$. As $m(\I', s') \leq m(\I', P'(s'))$,
we get that $P'(s')$ is a $r$-approximate solution to $\I'$, too.
Furthermore, since $P'(s')$ satisfies all constraints introduced in
step 1, we have $\opt(\I') - m(\I', P'(s')) = \opt(\I) - m(\I,
P(s'))$.  Let $\beta = 1 + 2kp^2$ and note that
\begin{align}
\frac{\opt(\I)}{m(\I, G(\I, s'))}
&=        \frac{m(\I, P(s'))}{m(\I, G(\I, s'))} + \frac{\opt(\I') - m(\I', P'(s'))}{m(\I, G(\I, s'))}                       &\leq \notag \\
&\leq 1 + \frac{\opt(\I') - m(\I', P'(s'))}{m(\I, G(\I, s'))} &\leq \notag \\
&\leq 1 + c \cdot \frac{\opt(\I') - m(\I', P'(s'))}{\opt(\I)} &\leq \notag \\
&\leq 1 + c\beta \cdot \frac{\opt(\I') - m(\I', P'(s'))}{\opt(\I')} &\leq \notag \\
&\leq 1 + c\beta \cdot \frac{\opt(\I') - m(\I', P'(s'))}{m(\I', P'(s'))} &\leq \notag \\
&\leq 1 + c\beta(r-1) . & \notag
\end{align}

\end{proof}

\begin{lemma} \label{lem:non-transitive}
  Let $H$ be a non-empty non-valid digraph such that
  \begin{itemize}
  \item $|V[H]| > 2$,
  \item $H$ is a core, and
  \item $H$ is not vertex-transitive.
  \end{itemize}
  Then, either (a) \prob{Max CSP$(\{H\})$-$B$} is hard to approximate,
  or (b) there exists a proper subset $X$ of $V$ such that $|X| \geq
  2$, $\proj{H}{X}$ is non-empty, $\proj{H}{X}$ is non-valid and for
  every $k$ there exists a $k'$ such that \prob{Max
    CSP$(\{\proj{H}{X}\})$-$k$} $\leq_{AP}$ \prob{Max
    CSP$(\{H\})$-$k'$}.
\end{lemma}
\begin{proof}
  We split the proof into three cases.

  \noindent \textbf{Case 1: There exists an orbit $\Omega_1 \subsetneq V[H]$ such
    that $\Omega_1^+$ contains at least one orbit.}

  If $\proj{H}{\Omega_1}$ is non-empty, then we get the result from
  Lemma~\ref{lem:orbit-rest} since $\Omega_1 \subsetneq V[H]$ (we
  cannot have $|\Omega_1| = 1$ because then $H$ would contain a loop).  Assume
  that $\proj{H}{\Omega_1}$ is empty.  As $\proj{H}{\Omega_1}$ is
  empty, we get that $\Omega_1^+$ is a proper subset of $V[H]$. If
  $\proj{H}{\Omega_1^+}$ is non-empty, then we get the result from
  Lemmas~\ref{lem:orbit+-impl}, \ref{lem:strict}
  and~\ref{lem:domrest}. Hence, we assume that $\proj{H}{\Omega_1^+}$
  is empty.

  Arbitrarily choose an orbit $\Omega_2 \subseteq \Omega_1^+$ and note
  that $\Omega_1^+ \cap \Omega_2^- = \emptyset$ since
  $\proj{H}{\Omega_1^+}$ is empty. If $\Omega_1^+ \cup \Omega_2^-
  \subsetneq V[H]$, then we get the result from
  Lemmas~\ref{lem:orbit+-impl}, \ref{lem:strict},
  \ref{lem:unary-union} and~\ref{lem:domrest} because
  $\proj{H}{\Omega_1^+ \cup \Omega_2^-}$ is non-empty. Hence, we can
  assume without loss of generality that $\Omega_1^+ \cup \Omega_2^- =
  V[H]$, and since $\Omega_1^+ \cap \Omega_2^- = \emptyset$, we have an
  partition of $V[H]$ into the sets $\Omega_1^+$ and $\Omega_2^-$.
  Using the same argument as for $\Omega_1^+$, we can assume that
  $\proj{H}{\Omega_2^-}$ is empty.
  Therefore, $\Omega_1^+$,$\Omega_2^-$ is a partition of $V[H]$ and
  $\proj{H}{\Omega_1^+}$,$\proj{H}{\Omega_2^-}$ are both empty. This
  implies that $H$ is bipartite and we get the result from
  Lemma~\ref{lem:bipartite}.

  \noindent \textbf{Case 2: There exists an orbit $\Omega_1 \subset V[H]$ such
    that $\Omega_1^-$ contains at least one orbit.}

  This case is analogous to the previous case.

  \noindent \textbf{Case 3: For every orbit $\Omega \subseteq V[H]$, neither $\Omega^+$
    nor $\Omega^-$ contains any orbits.}

  Pick any two orbits $\Omega_1$ and $\Omega_2$ (not necessarily
  distinct). Assume that there are $x \in \Omega_1$ and $y \in
  \Omega_2$ such that $(x, y) \in E[H]$.  Let $z$ be an arbitrary
  vertex in $\Omega_2$. Since $\Omega_2$ is an orbit of $H$, there is an
  automorphism $\rho \in \Aut(H)$ such that $\rho(y) = z$, so
  $(\rho(x), z) \in E[H]$. Furthermore, $\Omega_1$ is an orbit of
  $\Aut(H)$ so $\rho(x) \in \Omega_1$.  Since $z$ was chosen
  arbitrarily, we conclude that $\Omega_2 \subseteq \Omega_1^+$.
  However, this contradicts our assumption that neither $\Omega_1^+$
  nor $\Omega_1^-$ contains any orbit. We conclude that for any pair
  $\Omega_1$, $\Omega_2$ of orbits and any $x \in \Omega_1$, $y \in
  \Omega_2$, we have $(x, y) \not \in E[G]$.  This implies that $H$ is
  empty and Case~3 cannot occur.
\end{proof}

\begin{lemma} \label{lem:digraph}
  Let $H$ be a non-empty non-valid digraph. Then, \prob{Max
    CSP$(\{H\})$-$B$} is hard to approximate.
\end{lemma}
\begin{proof}
  Due to Lemmas~\ref{lem:core} and~\ref{lem:ap-red}, we can assume that
  $H$ is a core. If $H$ is vertex-transitive, then the result follows
  from Lemma~\ref{lem:trans-digraph}. If $H$ is not vertex-transitive,
  then we can obtain, by Lemma~\ref{lem:non-transitive}, a smaller
  graph $G$ such that $G$ has at least two vertices, $G$ is
  non-empty, $G$ is non-valid, and \prob{Max CSP$(G)$-$B$}
  $\leq_{AP}$ \prob{Max CSP$(H)$-$B$}. By repeatedly using
  Lemma~\ref{lem:non-transitive}, we will eventually obtain either a
  graph which is vertex-transitive graph or a proof of approximation
  hardness. In the former case the result follows from
  Lemma~\ref{lem:trans-digraph}.
\end{proof}

\subsection{Main Result} \label{sec:main-th}
Armed with the previous lemmas, it is sufficient to provide an arity
reduction argument (Lemma~\ref{lem:arity} below) and assemble the
various pieces to prove the main theorem. Lemma~\ref{lem:arity} was first
proved in~\cite{maxCSP-three} but we repeat the proof here to make
this report more self-contained.

\begin{lemma} \label{lem:arity}
  If $R$ is a non-empty non-valid relation of arity $n \geq 2$,
  then $R$ strictly implements a binary non-empty non-valid
  relation.
\end{lemma}
\begin{proof}
  We prove the lemma by induction on the arity of $R$. The result
  trivially holds for $n = 2$. Assume that the result holds for $n =
  k$, $k \geq 2$. We show that it holds for $n=k+1$.  Assume first
  that there exists $(a_1,\ldots,a_{k+1}) \in D^{k+1}$ such that
  $R(a_1,\ldots,a_{k+1}) = 1$ and $|\{a_1,\ldots,a_{k+1}\}| \leq k$.
  We assume without loss of generality that $a_k=a_{k+1}$ and consider
  the predicate $R'(x_1\zd x_k) = R(x_1,\ldots,x_k,x_k)$.  Note that
  this is a strict 1-implementation and $R'(d,\ldots,d) = 0$ for all
  $d \in D$. Furthermore, note that $R'$ is non-empty since
  $R'(a_1,\ldots,a_k)=1$.

  Assume now that $|\{a_1,\ldots,a_{k+1}\}| = k+1$ whenever
  $R(a_1,\ldots,a_{k+1})=1$. Consider the predicate $R'(x_1, \ldots,
  x_k)= \max_{y} R(x_1, \ldots, x_k, y)$, and note that this is a
  strict 1-implementation. We see that $R'(d,\ldots,d)=0$ for all $d
  \in D$ (due to the condition above) and $R'$ is non-empty since
  $R$ is non-empty.
\end{proof}

We are finally able to state the proof of the main theorem of this
section, Theorem~\ref{th:single}.

\begin{proof}
  Let $R$ be a relation in $R_D^{(n)}$. Clearly, \prob{Max
    CSP$(\{R\})$} can be solved in polynomial time if $R$ is valid. If
  $R$ is empty, then all solutions have the same
  measure.

  Otherwise, if $R$ is non-empty and not valid, then we can, due to
  Lemma~\ref{lem:arity}, strictly implement a binary relation $R'$
  with $R$ such that $R'$ is neither valid nor empty. Together with
  Lemma~\ref{lem:digraph} and Lemma~\ref{lem:strict}, we get the
  desired result.
\end{proof}

We will now give a simple example on how Theorem~\ref{th:single} can
be used for studying the approximability of constraint
languages. Consider the following observation: Let $\Gamma$ be a
constraint language, $R \in \Gamma$ and $\Omega$ an orbit in
$\Aut(\Gamma)$.  Then, $\proj{R}{\Omega}$ is either $d$-valid for
every $d \in \Omega$ or not $d$-valid for any $d \in \Omega$.

\begin{proposition} \label{prop:orbits}
  Let $O = \{\Omega \mid \mbox{$\Omega$ is an orbit in
    $\Aut(\Gamma)$}\}$ and let $\Gamma$ be a constraint language such
    that $O \subseteq \Gamma$.  If $\Gamma$ contains a $k$-ary,
    $k > 1$, relation $R$ that contains a tuple $(t_1,\ldots,t_k)$
    such that $R$ is not $t_i$-valid, for any $1 \leq i \leq k$, then
    \prob{Max CSP$(\Gamma)$} is hard to approximate.
\end{proposition}
\begin{proof}
  We can view the unary relation
  \[
  U = \bigcup \{\Omega \in O \mid t_i \in \Omega \mbox{ for some } 1 \leq i \leq k \}
  \]
  as a member of $\Gamma$ due to Lemma~\ref{lem:unary-union}. Now,
  $\proj{R}{U}$ is a non-empty, non-valid relation and
  approximability hardness follows from
  Lemmas~\ref{lem:strict},~\ref{lem:domrest}, and
  Theorem~\ref{th:single}.
\end{proof}

\begin{corollary} \label{cor:gamma-trans}
  Let $\Gamma$ be a constraint language such that $\Aut(\Gamma)$
  contains a single orbit.  If $\Gamma$ contains a non-empty
  $k$-ary, $k > 1$, relation $R$ which is not $d$-valid for all $d \in
  D$, then \prob{Max CSP$(\Gamma)$} is hard to approximate. Otherwise,
  \prob{Max CSP$(\Gamma)$} is tractable.
\end{corollary}
\begin{proof}
  If a relation $R$ with the properties described above exists, then
  \prob{Max CSP$(\Gamma)$} is hard to approximate by
  Proposition~\ref{prop:orbits} (note that $R$ cannot be $d$-valid for
  any $d$). Otherwise, every $k$-ary, $k>1$, relation $S \in \Gamma$
  is $d$-valid for all $d \in D$. If $\Gamma$ contains a unary
  relation $U$ such that $U \subsetneq D$, then $\Aut(\Gamma)$ would
  contain at least two orbits which contradict our assumptions. It
  follows that \prob{Max CSP$(\Gamma)$} is trivially solvable.
\end{proof}

Note that the constraint languages considered in
Corollary~\ref{cor:gamma-trans} may be seen as a generalisation of
vertex-transitive graphs.

\subsection{{\sc Max CSP} and Supermodularity} \label{sec:impl}
In this section, we will prove two results whose proofs make use of
Theorem~\ref{th:single}. The first result
(Proposition~\ref{prop:poset}) concerns the hardness of
approximating \prob{Max CSP$(\Gamma)$} for $\Gamma$ which contains
all at most binary relations which are 2-monotone (see
Section~\ref{sec:impl-prel} for a definition) on some partially ordered
set which is not a lattice order. The other result,
Theorem~\ref{th:2mon-lattice}, states that \prob{Max CSP$(\Gamma)$}
is hard to approximate if $\Gamma$ contains all at most binary
supermodular predicates on some lattice and in addition contains at
least one predicate which is not supermodular on the lattice.

These results strengthens earlier published
results~\cite{maxcsp-diamonds-tr,maxcsp-diamonds} in various ways
(e.g., they apply to a larger class of constraint languages or they
give approximation hardness instead of \cc{NP}-hardness). In
Section~\ref{sec:impl-prel} we give a few preliminaries which are needed
in this section while the new results are contained in
Section~\ref{sec:impl-results}.

\subsubsection{Preliminaries} \label{sec:impl-prel}
Recall that a partial order $\lattleq$ on a domain $D$ is a
\emph{lattice order} if, for every $x, y \in D$, there exist a
greatest lower bound $x \glb y$ and a least upper bound $x \lub y$.
The algebra $\mathcal{L} = (D; \glb, \lub)$ is a \emph{lattice}, and
$x \lub y = y \iff x \glb y = x \iff x \lattleq y$. We will write $x
\sqsubset y$ if $x \neq y$ and $x \sqsubseteq y$. All lattices we
consider will be finite, and we will simply refer to these algebras as
\emph{lattices} instead of using the more appropriate term
\emph{finite lattices}. The \emph{direct product} of $\mathcal{L}$,
denoted by $\mathcal{L}^n$, is the lattice with domain $D^n$ and
operations acting componentwise.

\begin{definition}[Supermodular function]
Let $\mathcal{L}$ be a lattice on $D$. A function $f : D^n \rightarrow
\mathbb{R}$ is called \emph{supermodular} on $\mathcal{L}$ if it satisfies,
\begin{align}
f(\tup{a}) + f(\tup{b}) \leq f(\tup{a} \sqcap \tup{b}) + f(\tup{a} \sqcup \tup{b})  \label{eq:supermod}
\end{align}
for all $\tup{a}, \tup{b} \in D^n$.
\end{definition}

The set of all supermodular predicates on a lattice $\mathcal{L}$ will
be denoted by $\Spmod_\mathcal{L}$ and a constraint language $\Gamma$ is
said to be supermodular on a lattice $\mathcal{L}$ if $\Gamma
\subseteq \Spmod_\mathcal{L}$. We will sometimes use an alternative
way of characterising supermodularity:

\begin{theorem}[\cite{greedyalg-posets-submod}] \label{th:spmod}
  An $n$-ary function $f$ is supermodular on a lattice $\mathcal{L}$
  if and only if it satisfies inequality \eqref{eq:supermod} for all
  $\tup{a} = (a_1, a_2, \ldots, a_n), \tup{b} = (b_1, b_2, \ldots,
  b_n) \in \mathcal{L}^n$ such that
  \begin{enumerate}
  \item $a_i = b_i$ with one exception, or
  \item $a_i = b_i$ with two exceptions, and, for each $i$, the
    elements $a_i$ and $b_i$ are comparable in $\mathcal{L}$.
  \end{enumerate}
\end{theorem}

The following definition first occurred in~\cite{supmod-maxCSP}.
\begin{definition}[Generalised 2-monotone]
  Given a poset $\mathcal{P}=(D,\lattleq)$, a predicate $f$ is
  said to be \emph{generalised 2-monotone on $\mathcal{P}$} if
  \[
  f(\tup{x}) = 1 \iff ((x_{i_1} \lattleq a_{i_1}) \land \ldots \land (x_{i_s} \lattleq a_{i_s})) \lor
                      ((x_{j_1} \lattgeq b_{j_1}) \land \ldots \land (x_{j_s} \lattgeq b_{j_s}))
  \]
  where $\tup{x} = (x_1, x_2, \ldots, x_n)$ and $a_{i_1}, \ldots,
  a_{i_s}, b_{j_1}, \ldots, b_{j_s} \in D$, and either of the two
  disjuncts may be empty.
\end{definition}

It is not hard to verify that generalised 2-monotone predicates on
some lattice are supermodular on the same lattice. For brevity,
we will use the
word \emph{2-monotone} instead of generalised 2-monotone.

The following theorem follows from~\cite[Remark~4.7]{maxCSP-fixed}.
The proof in~\cite{maxCSP-fixed} uses the corresponding unbounded
occurrence case as an essential stepping stone;
see~\cite{boolean-csp} for a proof of this latter result.
\begin{theorem}[\prob{Max CSP} on a Boolean domain] \label{th:maxCSP-bool}
  Let $D = \{0,1\}$ and $\Gamma \subseteq R_D$ be a core. If $\Gamma$
  is not supermodular on any lattice on $D$, then \prob{Max
    CSP$(\Gamma)$-$B$} is hard to approximate. Otherwise, \prob{Max
    CSP$(\Gamma)$} is tractable.
\end{theorem}

\subsubsection{Results} \label{sec:impl-results}

The following proposition is a combination of results proved
in~\cite{supmod-maxCSP} and~\cite{maxcsp-diamonds-tr}.
\begin{proposition} $~$
\begin{itemize}
\item If $\Gamma$ consists of 2-monotone relations on a lattice, then
  \prob{Max CSP$(\Gamma)$} can be solved in polynomial time.

\item Let $\mathcal{P}=(D,\lattleq)$ be a poset which is not a
  lattice.  If $\Gamma$ contains all at most binary 2-monotone
  relations on $\mathcal{P}$, then \prob{Max CSP$(\Gamma)$} is
  \cc{NP}-hard.
\end{itemize}
\end{proposition}

We strengthen the second part of the above result as follows:
\begin{proposition} \label{prop:poset}
  Let $\lattleq$ be a partial order, which is not a lattice order, on
  $D$. If $\Gamma$ contains all at most binary 2-monotone relations on
  $\lattleq$, then \prob{Max CSP$(\Gamma)$-$B$} is hard to approximate.
\end{proposition}
\begin{proof}
  Since $\lattleq$ is a non-lattice partial order, there exist two
  elements $a, b \in D$ such that either $a \glb b$ or $a \lub b$ do
  not exist.  We will give a proof for the first case and the other
  case can be handled analogously.

  Let $g(x, y) = 1 \iff (x \lattleq a) \land (y \lattleq b)$. The
  predicate $g$ is 2-monotone on $\mathcal{P}$ so $g \in \Gamma$. We
  have two cases to consider: (a) $a$ and $b$ have no common lower
  bound, and (b) $a$ and $b$ have at least two maximal common lower
  bounds. In the first case $g$ is not valid. To see this, note that
  if there is an element $c \in D$ such that $g(c,c) = 1$, then $c
  \lattleq a$ and $c \lattleq b$, and this means that $c$ is a common
  lower bound for $a$ and $b$, a contradiction.
  Hence, $g$ is not valid, and the proposition
  follows from Theorem~\ref{th:single}.

  In case (b) we will use the domain restriction technique from
  Lemma~\ref{lem:domrest} together with Theorem~\ref{th:single}. In
  case (b), there exist two distinct elements $c, d \in D$, such that
  $c, d \lattleq a$ and $c, d \lattleq b$. Furthermore, we can assume
  that there is no element $z \in D$ distinct from $a,b,c$ such that
  $c \lattleq z \lattleq a, b$, and, similarly, we can assume there is
  no element $z' \in D$ distinct from $a,b,d$ such that $d \lattleq z'
  \lattleq a, b$.

  Let $f(x) = 1 \iff (x \lattgeq c) \land (x \lattgeq d)$. This
  predicate is 2-monotone on $\mathcal{P}$. Note that there is no
  element $z \in D$ such that
  $f(z) = 1$ and $g(z, z) = 1$, but we have $f(a) = f(b) = g(a, b) = 1$.  By
  restricting the domain to $D' = \{x \in D \mid f(x) = 1\}$ with
  Lemma~\ref{lem:domrest}, the result follows from
  Theorem~\ref{th:single}.
\end{proof}

A \emph{diamond} is a lattice $\mathcal{L}$ on a domain $D$ such
that $|D|-2$ elements are pairwise incomparable. That is, a diamond
on $|D|$ elements consist of a top element, a bottom element and
$|D|-2$ elements which are pairwise incomparable. The following
result was proved in~\cite{maxcsp-diamonds}.
\begin{theorem} \label{th:2mon-diam}
  Let $\Gamma$ contain all at most binary 2-monotone predicates on
  some diamond $\mathcal{L}$. If $\Gamma \not \subseteq
  \Spmod_{\mathcal{L}}$, then \prob{Max CSP$(\Gamma)$} is \cc{NP}-hard.
\end{theorem}

By modifying the original proof of Theorem~\ref{th:2mon-diam}, we can
strengthen the result in three ways: our result applies to
arbitrary lattices, we prove inapproximability results instead of
\cc{NP}-hardness, and we prove the result for bounded occurrence
instances.

\begin{theorem} \label{th:2mon-lattice}
  Let $\Gamma$ contain all at most binary 2-monotone predicates on an
  arbitrary lattice $\mathcal{L}$. If $\Gamma \not \subseteq
  \Spmod_{\mathcal{L}}$, then \prob{Max CSP$(\Gamma)$-$B$} is
  hard to approximate.
\end{theorem}
\begin{proof}
  Let $f \in \Gamma$ be a predicate such that $f \not \in
  \Spmod_{\mathcal{L}}$. We will first prove that $f$ can be assumed
  to be at most binary. By Theorem~\ref{th:spmod}, there is a unary or
  binary predicate $f' \not \in \Spmod_{\mathcal{L}}$ which can be
  obtained from $f$ by substituting all but at most two variables by
  constants. We present the initial part of the proof with the
  assumption that $f'$ is binary and the case when $f'$ is unary can
  be dealt with in the same way. Denote the constants by $a_3, a_4,
  \ldots, a_n$ and assume that $f'(x, y) = f(x, y, a_3, a_4, \ldots,
  a_n)$.

  Let $k \geq 5$ be an integer and assume that \prob{Max CSP$(\Gamma
    \cup \{f'\})$-$k$} is hard to approximate. We will prove that
  \prob{Max CSP$(\Gamma)$-$k$} is hard to approximate by exhibiting an
  $AP$-reduction from \prob{Max CSP$(\Gamma \cup \{f'\})$-$k$} to
  \prob{Max CSP$(\Gamma)$-$k$}.  Given an instance $\I = (V, C)$ of
  \prob{Max CSP$(\Gamma \cup \{f'\})$-$k$}, where $C = \{C_1, C_2,
  \ldots, C_q\}$, we construct an instance $\I' = (V', C')$ of
  \prob{Max CSP$(\Gamma)$-$k$} as follows:
  \begin{enumerate}
  \item for any constraint $(f', \tup{v}) = C_j \in C$, introduce the constraint $(f,
    \tup{v'})$ into $C$, where $\tup{v'} = (v_1, v_2, y^j_3,$ $\ldots,
    y^j_n)$, and add the fresh variables $y^j_3, y^j_4, \ldots, y^j_n$
    to $V'$.  Add two copies of the constraints $y^j_i \lattleq a_i$
    and $a_i \lattleq y^j_i$ for each $i \in \{3, 4, \ldots, n\}$ to $C'$.

  \item for other constraints, i.e., $(g, \tup{v}) \in C$ where $g
    \neq f'$, add $(g, \tup{v})$ to $C'$.
  \end{enumerate}

  It is clear that $\I'$ is an instance of \prob{Max
    CSP$(\Gamma)$-$k$}. If we are given a solution $s'$ to $\I'$, we
  can construct a new solution $s''$ to $\I'$ by letting $s''(y^j_i) =
  a_i$ for all $i,j$ and $s''(x) = s'(x)$, otherwise. Denote this
  transformation by $P$, so $s'' = P(s')$. It is not hard
  to see that $m(\I', P(s')) \geq m(\I', s')$.

  From Lemma~\ref{lem:in-apx} we know that there is a constant $c$
  and polynomial-time $c$-approximation algorithm $A$ for \prob{Max
    CSP$(\Gamma \cup \{f'\})$}. We construct the algorithm $G$ in the
  $AP$-reduction as follows:
\[
G(\I, s') =
\left\{ \begin{array}{ll}
\proj{P(s')}{V} & \textrm{if } m(\I, \proj{P(s')}{V}) \geq m(\I, A(\I)), \\
A(\I)           & \textrm{otherwise.}
\end{array} \right.
\]
We see that $\opt(\I) / m(\I, G(\I, s')) \leq c$.

By Lemma~\ref{lem:in-apx}, there is a constant $c'$ such that for
any instance $\I$ of \prob{Max CSP$(\Gamma)$}, we have $\opt(\I) \geq
c'|C|$. Furthermore, due to the construction of $\I'$ and the fact
that $m(\I', P(s')) \geq m(\I', s')$, we have
\begin{align} \notag
\opt(\I') &\leq \opt(\I) + 4(n-2)|C| \\ \notag
          &\leq \opt(\I) + \frac{4(n-2)}{c'} \cdot \opt(\I) \\ \notag
          &\leq \opt(\I) \cdot \left(1 + \frac{4(n-2)}{c'} \right) . \notag
\end{align}

Let $s'$ be an $r$-approximate solution to $\I'$. As $m(\I', s')
\leq m(\I', P(s'))$, we get that $P(s')$ also is an $r$-approximate
solution to $\I'$.  Furthermore, since $P(s')$ satisfies all
constraints introduced in step 1, we have $\opt(\I') - m(\I', P(s'))
= \opt(\I) - m(\I, \proj{P(s')}{V})$.  Let $\beta = 1 + 4(n-2)/c'$
and note that
\begin{align}
& \ \ \ \ \frac{\opt(\I)}{m(\I, G(\I, s'))} & = \notag \\
&=    \frac{m(\I, \proj{P(s')}{V})}{m(\I, G(\I, s'))} + \frac{\opt(\I') - m(\I', P(s'))}{m(\I, G(\I, s'))}                       &\leq & & \notag \\
&\leq 1 + \frac{\opt(\I') - m(\I', P(s'))}{m(\I, G(\I, s'))} &\leq& \ 1 + c \cdot \frac{\opt(\I') - m(\I', P(s'))}{\opt(\I)}           &\leq \notag \\
&\leq 1 + c\beta \cdot \frac{\opt(\I') - m(\I', P(s'))}{\opt(\I')} &\leq& \ 1 + c\beta \cdot \frac{\opt(\I') - m(\I', P(s'))}{m(\I', P(s'))} &\leq \notag \\
&\leq 1 + c\beta(r-1) . & & \notag
\end{align}

We conclude that \prob{Max CSP$(\Gamma)$-$k$} is hard to approximate
if \prob{Max CSP$(\Gamma \cup \{f'\})$-$k$} is hard to approximate.

  We will now prove that \prob{Max CSP$(\Gamma)$-$B$} is hard to
  approximate under the assumption that $f$ is at most binary.  We say
  that the pair $(\tup{a}, \tup{b})$ \emph{witnesses the
    non-supermodularity of $f$} if $f(\tup{a}) + f(\tup{b}) \not \leq
  f(\tup{a} \glb \tup{b}) + f(\tup{a} \lub \tup{b})$.

\noindent \textbf{Case 1: $f$ is unary.} As $f$ is not supermodular
on $\mathcal{L}$, there exists elements $a, b \in \mathcal{L}$ such
that $(a, b)$ witnesses the non-supermodularity of $f$.

Note that $a$ and $b$ cannot be comparable because we would have
$\{a \lub b, a \glb b\} = \{a, b\}$, and so $f(a \lub b) + f(a \glb
b) = f(a) + f(b)$ contradicting the choice of $(a,b)$. We can now
assume, without loss of generality, that $f(a) = 1$. Let
$z_* = a \glb b$ and $z^* = a \lub b$. Note that the two predicates
$u(x) = 1 \iff x \lattleq z^*$ and $u'(x) = 1 \iff z_* \lattleq x$
are 2-monotone and, hence, contained in $\Gamma$. By using
Lemma~\ref{lem:domrest}, it is therefore enough to prove
approximation hardness for \prob{Max CSP$(\proj{\Gamma}{D'})$-$B$},
where $D' = \{x \in D \mid z_* \lattleq x \lattleq z^* \}$.

\noindent \textbf{Subcase 1a: $f(a) = 1$ and $f(b) = 1$.}
At least one of $f(z^*) = 0$ and $f(z_*) = 0$ must hold.

Assume that $f(z_*) = 0$, the other case can be handled in a similar
way. Let $g(x, y) = 1 \iff [(x \sqsubseteq a) \land (y \sqsubseteq
b)]$ and note that $g$ is 2-monotone so $g \in \Gamma$.

Let $d$ be an arbitrary element in $D'$ such that $g(d, d) = 1$.
From the definition of $g$ we know that $d \lattleq a, b$ so $d
\lattleq z_*$ which implies that $d = z_*$. Furthermore, we have
$g(a, b) = 1, f(a) = f(b) = 1$, and $f(z_*) = 0$. Let $D''=\{x\in
D'\mid f(x)=1\}$. By applying Theorem~\ref{th:single} to $g|_{D''}$,
we see that \prob{Max CSP$(\proj{\Gamma}{D''})$-$B$} is hard to
approximate. Now Lemma~\ref{lem:domrest} implies the result for
\prob{Max CSP$(\proj{\Gamma}{D'})$-$B$}, and hence for \prob{Max
CSP$(\Gamma)$-$B$}.

\noindent \textbf{Subcase 1b: $f(a) = 1$ and $f(b) = 0$.}
In this case, $f(z^*) = 0$ and $f(z_*) = 0$ holds.

If there exists $d \in D'$ such that $b \lattl d \lattl z^*$ and
$f(d) = 1$, then we get $f(a) = 1$, $f(d) = 1$, $a \lub d = z^*$ and
$f(z^*) = 0$, so this case can be handled by Subcase~1a. Assume that
such an element $d$ does not exist.

Let $u(x) = 1 \iff b \sqsubseteq x$. The predicate $u$ is 2-monotone
so $u \in \Gamma$. Let $h(x) = f|_{D'}(x) + u|_{D'}(x)$. By the
observation above, this is a strict implementation. By
Lemmas~\ref{lem:strict} and~\ref{lem:ap-red}, it is sufficient to
prove the result for $\Gamma'=\Gamma|_{D'} \cup\{ h\}$. This can be
done exactly as in the previous subcase, with $D''=\{x\in D'\mid
h(x)=1\}$.


\noindent \textbf{Case 2: $f$ is binary.} We now assume that Case~1
does not apply. By Theorem~\ref{th:spmod}, there exist $a_1, a_2,
b_1, b_2$ such that
\begin{align} \label{eq:non-spmod}
  f(a_1, a_2) + f(b_1, b_2) \not \leq f(a_1 \lub b_1, a_2 \lub b_2) +
  f(a_1 \glb b_1, a_2 \glb b_2)
\end{align}
where $a_1, b_1$ are comparable and $a_2, b_2$ are comparable. Note
that we cannot have $a_1 \lattleq b_1$ and $a_2 \lattleq b_2$,
because then the right hand side of~\eqref{eq:non-spmod} is equal to
$f(b_1, b_2) + f(a_1, a_2)$ which is a contradiction. Hence, we can
without loss of generality assume that $a_1 \lattleq b_1$ and $b_2
\lattleq a_2$.

As in Case~1, we will use Lemma~\ref{lem:domrest} to restrict our
domain. In this case, we will consider the subdomain $D' = \{ x \in D
\mid z_* \lattleq x \lattleq z^* \}$ where $z_* = a_1 \glb b_2$ and
$z^* = a_2 \lub b_1$. As the two predicates $u_{z^*}(x)$ and
$u_{z_*}(x)$, defined by $u_{z^*}(x) = 1 \iff x \lattleq z^*$ and
$u_{z_*}(x) = 1 \iff z_* \lattleq x$, are 2-monotone predicates and
members of $\Gamma$, Lemma~\ref{lem:domrest} tells us that it is
sufficient to prove hardness for \prob{Max CSP$(\Gamma')$-$B$} where
$\Gamma' = \proj{\Gamma}{D'}$.

We define the functions $t_i : \{0,1\} \rightarrow \{a_i, b_i\}$,
$i=1,2$ as follows:
\begin{itemize}
\item $t_1(0) = a_1$ and $t_1(1) = b_1$;
\item $t_2(0) = b_2$ and $t_2(1) = a_2$.
\end{itemize}
Hence, $t_i(0)$ is the least element of $a_i$ and $b_i$ and $t_i(1)$
is the greatest element of $a_i$ and $b_i$.

Our strategy will be to reduce a certain Boolean \prob{Max CSP}
problem to \prob{Max CSP$(\Gamma')$-$B$}. Define three Boolean
predicates as follows: $g(x, y) = f(t_1(x), t_2(y))$, $c_0(x) = 1
\iff x = 0$, and $c_1(x) = 1 \iff x = 1$. One can verify that
\prob{Max CSP$(\{c_0, c_1, g\})$-$B$} is hard to approximate for
each possible choice of $g$, by using Theorem~\ref{th:maxCSP-bool};
consult Table~\ref{tab:g} for the different possibilities of $g$.

\begin{table}
\caption{Possibilities for $g$.} \label{tab:g}
\begin{center}
\begin{tabular}{cc|cc|ccccc}
$x$ & $y$ & $t_1(x)$ & $t_2(y)$ & \multicolumn{5}{c}{$g(x, y)$} \\
\hline
$0$ & $0$ & $a_1$    & $b_2$    & $0$ & $0$ & $0$ & $0$ & $1$ \\
$0$ & $1$ & $a_1$    & $a_2$    & $1$ & $1$ & $0$ & $1$ & $1$ \\
$1$ & $0$ & $b_1$    & $b_2$    & $1$ & $0$ & $1$ & $1$ & $1$ \\
$1$ & $1$ & $b_1$    & $a_2$    & $1$ & $0$ & $0$ & $0$ & $0$
\end{tabular}
\end{center}
\end{table}

The following 2-monotone predicates (on $D'$) will be used in the
reduction:
\[
h_i(x, y) = 1 \iff [(x \sqsubseteq z_*) \land (y \sqsubseteq
t_i(0))] \lor [(z^* \sqsubseteq x) \land (t_i(1) \sqsubseteq y)],
i=1,2.
\]
The predicates $h_1, h_2$ are 2-monotone so they belong to $\Gamma'$.
We will also use the following predicates:
\begin{itemize}
\item $L_d(x) = 1 \iff x \sqsubseteq d$,
\item $G_d(x) = 1 \iff d \sqsubseteq x$, and
\item $N_{d, d'}(x) = 1 \iff (x \sqsubseteq d) \lor (d' \sqsubseteq x)$
\end{itemize}
for arbitrary $d, d' \in D'$. These predicates are 2-monotone.

Let $w$ be an integer such that \prob{Max CSP$(\{g,c_0,c_1\})$-$w$}
is hard to approximate; such an integer exists according to
Theorem~\ref{th:maxCSP-bool}. Let $\I = (V, C)$, where $V = \{x_1,
x_2, \ldots, x_n\}$ and $C = \{C_1, \ldots, C_m\}$, be an instance
of \prob{Max CSP$(\{g, c_0, c_1\})$-$w$}.  We will construct an
instance $\I'$ of \prob{Max CSP$(\Gamma')$-$w'$}, where $w' = 8w+5$,
as follows:
\begin{enumerate}
\item[1.] For every $C_i \in C$ such that $C_i = g(x_j, x_k)$,
  introduce
\begin{enumerate}
\item two fresh variables $y_j^i$ and $y_k^i$,
\item the constraint $f(y_j^i, y_k^i)$,
\item $2w+1$ copies of the constraints $L_{b_1}(y_j^i), G_{a_1}(y_j^i), N_{a_1, b_1}(y_j^i)$,
\item $2w+1$ copies of the constraints $L_{a_2}(y_k^i), G_{b_2}(y_k^i), N_{b_2, a_2}(y_k^i)$, and
\item $2w+1$ copies of the constraints $h_1(x_j, y_j^i), h_2(x_k, y_k^i)$.
\end{enumerate}
\item[2.] for every $C_i \in C$ such that $C_i = c_0(x_j)$, introduce the
  constraint $L_{z_*}(x_j)$, and
\item[3.] for every $C_i \in C$ such that $C_i = c_1(x_j)$, introduce the
  constraint $G_{z^*}(x_j)$.
\end{enumerate}

The intuition behind this construction is as follows: due to the
bounded occurrence property and the quite large number of copies of
the constraints in steps 1c, 1d and 1e, all of those constraints will
be satisfied in ``good'' solutions. The elements $0$ and $1$ in the
Boolean problem corresponds to $z_*$ and $z^*$, respectively. This may
be seen in the constraints introduced in steps 2 and 3. The
constraints introduced in step 1c essentially force the variables
$y_j^i$ to be either $a_1$ or $b_1$, and the constraints in step 1d work
in a similar way. The constraints in step 1e work as bijective
mappings from the domains $\{a_1, b_1\}$ and $\{a_2, b_2\}$ to $\{z_*,
z^*\}$. For example, $h_1(x_j, y_j^i)$ will set $x_j$ to $z_*$ if
$y_j^i$ is $a_1$, otherwise if $y_j^i$ is $b_1$, then $x_j$ will be
set to $z^*$. Finally, the constraint introduced in step 1b
corresponds to $g(x_j, x_k)$ in the original problem.

It is clear that $\I'$ is an instance of \prob{Max
  CSP$(\Gamma')$-$w'$}. Note that due to the bounded occurrence
  property of $\I'$, a solution which does not satisfy all constraints
  introduced in steps 1c, 1d and 1e can be used to construct a new
  solution which satisfies those constraints and has a measure which
  is greater than or equal to the measure of the original solution. We
  will denote this transformation of solutions by $P$.

Given a solution $s'$ to $\I'$, we can construct a solution $s =
G(s')$ to $\I$ by, for every $x \in V$, letting $s(x) = 0$ if
$P(s')(x) = z_*$ and $s(x) = 1$, otherwise.

Let $M$ be the number of constraints in $C$ of type $g$. We have that,
for an arbitrary solution $s'$ to $\I'$, $m(\I', P(s')) = m(\I, G(s'))
+ 8(2w+1) \cdot M \geq m(\I', s')$.  Furthermore, $\opt(\I') = \opt(I)
+ 8(2w+1)M$.

Now, assume that $\opt(\I')/m(\I', s') \leq \epsilon'$. It follows that
$\opt(\I')/m(\I', P(s')) \leq \epsilon'$ and
\begin{align} \notag
&\frac{\opt(I) + 8(2w+1)M}{m(I, G(s')) + 8(2w+1)M} \leq \epsilon'   &\Rightarrow  \\ \notag
&\opt(I) \leq \epsilon' m(I, G(s')) + (\epsilon' - 1) 8(2w+1)M &\Rightarrow \\ \notag
&\frac{\opt(\I)}{m(\I, G(s'))} \leq \epsilon' + \frac{8(2w+1)M(\epsilon' - 1)}{m(\I, G(s'))} .\notag
\end{align}
Furthermore, by standard arguments, we can assume that $m(\I, G(s'))
\geq |C|/c$, for some constant $c$. We get,
\begin{align}
\frac{\opt(\I)}{m(\I, G(s'))} \leq \epsilon' + 8(2w+1)c(\epsilon' -
1) . \notag
\end{align}
Hence, a polynomial time approximation algorithm for \prob{Max
  CSP$(\Gamma')$-$w'$} with performance ratio $\epsilon'$ can be used
to obtain $\epsilon''$-approximate solutions, where
$\epsilon''$ is given by $\epsilon'+8(2w+1)c(\epsilon'-1)$, for \prob{Max
  CSP$(\{c_0, c_1, g\})$-$w$} in polynomial time. Note that
$\epsilon''$ tends to $1$ as $\epsilon'$
approaches $1$. This implies that \prob{Max CSP$(\Gamma')$-$w'$} is
hard to approximate because \prob{Max
  CSP$(\{c_0, c_1, g\})$-$w$} is hard to approximate.
\end{proof}

\section{Conclusions and Future Work} \label{sec:concl}
This report have two main results: the first one is that \prob{Max
  CSP$(\Gamma)$} has a hard gap at location 1 whenever $\Gamma$ satisfies
a certain condition which makes \prob{CSP$(\Gamma)$} \cc{NP}-hard.
This condition captures all constraint languages which are currently known to make
\prob{CSP$(\Gamma)$} \cc{NP}-hard. This condition has also been
conjectured to be the dividing line between tractable (in \cc{P})
\prob{CSP}s and \cc{NP}-hard \prob{CSP}s. The second result is that
single relation \prob{Max CSP} is hard to approximate except in a
few cases where optimal solutions can be found trivially.

It is possible to strengthen these results in a number of ways. The
following possibilities applies to both of our results.

We have paid no attention to the constant which we prove
inapproximability for. That is, given a constraint language
$\Gamma$, what is the smallest constant $c$ such that \prob{Max
CSP$(\Gamma)$} is not approximable within $c-\epsilon$ for any
$\epsilon > 0$ in polynomial time? For some relations a lot of work has
been done in this direction, cf.~\cite{Ausiello99:complexity} for
more details.


We have a constant number of variable occurrences in our hardness
results, but the constant is unspecified. For some problems, for
example \prob{Max 2Sat}, it is known that allowing only three variable
occurrences still makes the problem hard to approximate (even
\cc{APX}-hard)~\cite{Ausiello99:complexity}. This is also true for
some other \prob{Max CSP} problems such as \prob{Max
  Cut}~\cite{apx-cubic}. This leads to the questions: is \prob{Max
  CSP$(\{R\})$-$3$} hard to approximate for all non-valid non-empty
$R$? and, is it true that \prob{Max CSP$(\Gamma)$-$3$} has a hard gap
at location 1 whenever \prob{Max CSP$(\Gamma)$-$B$} has a hard gap at
location 1?

One of the main open problems is to classify \prob{Max
CSP$(\Gamma)$} for all constraint languages $\Gamma$, with respect
to tractability of finding an optimal solution. The current results
in this
direction~\cite{supmod-maxCSP,maxCSP-fixed,maxCSP-three,maxcsp-diamonds}
seems to indicate that the concept of \emph{supermodularity} is of
central importance for the complexity of \prob{Max CSP}. However,
the problem is open on both ends --- we do not know if supermodularity
implies tractability and neither do we know if non-supermodularity
implies non-tractability. Here ``non-tractability'' should be
interpreted as ``not in \cc{PO}'' under some suitable
complexity-theoretic assumption, the questions of \cc{NP}-hardness
and approximation hardness are, of course, also open.

\section*{Acknowledgements.}
The authors would like to thank Gustav Nordh for comments which have
improved the presentation of this paper. Peter Jonsson is partially
supported by the \emph{Center for Industrial Information Technology}
(CENIIT) under grant 04.01, and by the \emph{Swedish Research Council}
(VR) under grant 621-2003-3421.  Andrei Krokhin is supported by the UK
EPSRC grant EP/C543831/1.  Fredrik Kuivinen is supported by the
\emph{National Graduate School in Computer Science} (CUGS), Sweden.


\end{document}